\documentclass[11pt]{article}
\usepackage{amsmath,amssymb,amsfonts,latexsym,graphicx,amsthm}
\usepackage{fullpage,color}
\usepackage{url,hyperref}
\usepackage{comment}
\usepackage[linesnumbered,boxed,ruled,vlined]{algorithm2e}
\usepackage{framed}
\usepackage[shortlabels]{enumitem}
\usepackage{cleveref}
\usepackage{todonotes}
\usepackage[normalem]{ulem}

\usepackage{subcaption}
\usepackage[margin=1in]{geometry}
\usepackage{ifthen}
\usepackage{thm-restate}
\usepackage{xcolor}

\usepackage{stmaryrd}
\usepackage{lineno}

\newtheorem{theorem}{Theorem}
\newtheorem{lemma}{Lemma}
\newtheorem{proposition}{Proposition}
\newtheorem{definition}{Definition}
\newtheorem{corollary}{Corollary}

\newtheorem{observation}{Observation}

        {\medskip}

\makeatletter
\newcommand{\xrightsquigarrow}[2][]{\ext@arrow 0359\rightsquigarrowfill@{#1}{#2}}
\def\rightsquigarrowfill@{\arrowfill@\relax\relax\rightsquigarrow}
\makeatother

\newcommand{\eps}{\varepsilon}
\newcommand{\hide}[1]{}

\newcommand{\dist}{\mathrm{dist}}
\newcommand{\width}{\mathrm{width}}

\DeclareMathOperator*{\argmin}{arg\,min}

\DeclareMathOperator{\diam}{\mathrm{diam}}

\definecolor{BrickRed}{rgb}{0.8, 0.25, 0.33}

\begin{document}

	\title{Bichromatic Geometric Spanners\thanks{A preliminary version of this paper appeared in the proceedings of APPROX 2026 (\href{https://doi.org/10.4230/LIPIcs.APPROX/RANDOM.2026.6}{doi:10.4230/LIPIcs.APPROX/RANDOM.2026.6}). Research supported, in part, by the NSF award DMS-2154347.}}
    \author{Theodore Fung\thanks{Department of Mathematics, California State University Northridge, Los Angeles, CA, USA.} 
    \and Csaba D. T\'oth\thanks{Department of Mathematics, California State University Northridge, Los Angeles, CA; and Department of Computer Science, Tufts University, Medford, MA, USA.}}
	\date{}
	\maketitle

\begin{abstract}
For an edge-weighted graph $G=(V,E)$ and a stretch parameter $t\geq 1$, a $t$-spanner is a subgraph $H\subseteq G$ such that the shortest path distances in $G$ and $H$ satisfy $\delta_H(u,v)\leq t\, \delta_G(u,v)$ for all $u,v\in V$. In metric spanners, $V$ is a finite metric space, and $G$ is the complete graph with edge weights corresponding to the distances between the endpoints. When $G$ is the complete graph on $n$ points in the plane, $O(n)$-size $t$-spanners are possible for any $t>1$: For every $\eps>0$, there is a $(1+\eps)$-spanner with $O(n/\eps)$ edges (i.e., the stretch can be arbitrarily close to 1). 

When $G=K(R,B)$ is the complete bipartite graph on $n$ bichromatic points in the plane, in general, no spanner construction can guarantee stretch $t<3$ with $o(n^2)$ edges. Bose et al.~(SICOMP 2009) constructed a $(3+\eps)$-spanner with $O(n\log n)$ edges for any constant $\eps>0$. Our main result is a new construction for a $(3+\eps)$-spanner with $O(\sqrt{1/\eps}\cdot n)$ edges. Eliminating the $O(\log n)$ factor resolves a problem left open for more than 17 years, and raises a new research problem about optimizing the dependence on $\eps$. 
When off-metric Steiner points are allowed, no Steiner spanner construction can guarantee stretch $t<2$ with $o(n^2)$ edges, and we construct a bipartite Steiner $(2+\eps)$-spanner for $K(R,B)$ with $O(\sqrt{1/\eps}\cdot n)$ edges.

We also study (non-Steiner) spanners for $G=K(R,B)$ on $n$ bichromatic points on the real line: In this case, we show that the MST of $K(R,B)$ is a 7-spanner, and we construct a 3-spanner with at most $2n-3$ edges.
\end{abstract}

\section{Introduction}
\label{sec:intro}

\paragraph{Multichromatic Euclidean Spanners.} 
A $t$-spanner for an edge-weighted graph $G=(S,E)$  is a subgraph $H=(S,E')$ that contains, for every edge $ab\in E$, an $ab$-path of weight at most $t\cdot w(ab)$. 
Let $S$ be a set of $n$ bichromatic points in Euclidean plane, with coloring $c:S\to \{\mathrm{red},\mathrm{blue}\}$. Let $R$ and $B$, resp., denote the set of red and blue points in $S$, and let $K(R,B)$ denote the complete bipartite graph on the partite sets $R$ and $B$. The weight of an edge $ab$ is the Euclidean length $|ab|$ of the line segment $ab$.  
A $t$-spanner for the edge-weighted graph $K(R,B)$ is called, for short, a \emph{bichromatic $t$-spanner}.

For a set $S$ of $n$ (uncolored) points in the plane and any $\eps>0$, there exists a $(1+\eps)$-spanner with $O(n/\eps)$ edges~\cite{Clarkson87,Keil88}. 
Bose et al.~\cite{BoseCCMMS09} noted that for a 2-colored point set $S=R\cup B\subset \mathbb{R}^2$, in general, one cannot find a bichromatic $t$-spanner with $o(n^2)$ edges and stretch $t<3$. Specifically, for every $\eps>0$, one can arrange $\lceil n/2\rceil$ red and $\lfloor n/2\rfloor$ blue points in two clusters of diameter $\frac{\eps}{3}$ at unit distance apart: If a subgraph $H\subseteq K(R,B)$ misses any bichromatic edge $ab$, then every $ab$-path in $H$ has length at least $(3-\eps)|ab|$.
On the positive side, Bose et al.~\cite{BoseCCMMS09} constructed, for every 2-colored set $S=R\cup B\subset \mathbb{R}^2$ and constant $\eps>0$, a bichromatic $(5+\eps)$-spanner with $O(n)$ edges and a bichromatic $(3+\eps)$-spanner with  $O(n\log n)$ edges. 

\subsection{Our Results}
The main problem left open by Bose et al.~\cite{BoseCCMMS09} was whether a bichromatic $(3+\eps)$-spanner with $O(n)$ edges always exists. Our main result settles this problem in the affirmative. 

\begin{theorem}\label{thm:main}
For every $\eps\in(0,1)$ and every 2-colored set of $n$ points in the plane, one can construct a bichromatic $(3+\eps)$-spanner with $O(\sqrt{1/\eps}\cdot n)$ edges in $O(\sqrt{1/\eps}\cdot n\log n)$ time.
\end{theorem}
\Cref{thm:main} generalizes to multichromatic point sets in Euclidean $d$-space for constant dimension $d\geq 2$, and yields a $(3+\eps)$-spanner with $O(n / \eps^{(d-1)/2})$ edges. This bound does not depend on the number of colors $k\geq 2$. We describe the spanner construction and analyze it in detail in the plane, and then briefly sketch the generalization to higher dimensions. 

Since our bound on the size of the spanner is linear in $n$, it raises a new problem about the optimal tradeoff between $\eps>0$ and the size of a bichromatic $(3+\eps)$-spanner in the plane. We can show that the dependence on $\eps>0$ in \Cref{thm:main} is necessary. 
\begin{restatable}{theorem}{lowerbound}
\label{thm:LB}
For every $\eps\in (0,\frac19)$ and $n= \Omega(\sqrt{1/\eps})$, there is a 2-colored set of $n$  points in the plane for which every bichromatic $(3+\eps)$-spanner has $\Omega(\log \log \eps^{-1}\cdot n)$ edges. 
\end{restatable}
The 2-colored point set in \Cref{thm:LB} consists of red and blue points on opposite sides of unit squares. We show that  \Cref{thm:LB} is tight for this construction (cf.~\Cref{pro:tight}).

\paragraph{Steiner Spanners.} 
A \emph{Steiner $t$-spanner} for an edge-weighted graph $G=(S,E)$ is a graph $H=(\widehat{S},\widehat{E})$ such that $S\subseteq \widehat{S}$ and for every point pair $\{a,b\}\subset S$, the shortest-path distances $d_G$ and $d_H$ satisfy $d_G(a,b)\leq d_H(a,b)\leq t\cdot  d_G(a,b)$. Note that (1) the Steiner point in $\widehat{S}\setminus S$ need not be points in the plane, (2) the edge weights in $\widehat{E}$ are not necessarily Euclidean distances, and (3) the lower bound $d_G(a,b)\leq d_H(a,b)$ means that the shortest path distance in $H$ is noncontracting\footnote{If the noncontracting condition $d_G(a,b)\leq d_H(a,b)$ is dropped, one could easily construct a $(1+\eps)$-spanner of size $O(n/\eps)$, even if the Steiner points are embedded in the plane and the edge weights are Euclidean distances: It is enough to add a red (resp., blue) Steiner point next to each blue (resp., red) input point, and compute an uncolored $(1+\eps)$-spanner on the resulting red-blue pairs.}. We consider Steiner spanners for the complete bipartite graph $K(R,B)$ for partite sets $R,B\subset \mathbb{R}^2$ with Euclidean edge weights. We show that, in general, a sparse Steiner spanner cannot achieve stretch less than 2. 
\begin{restatable}{theorem}{SteinerLB}
\label{thm:SteinerLB}
For every $\eps\in (0,\frac19)$ and every even integer $n\geq 2$, there is a 2-colored set $S=R\cup B$ of $n$ points in the plane such that every Steiner $(2-\eps)$-spanner for $K(R,B)$ has $\Omega(n^2)$ edges. 
\end{restatable}

We also construct, for every $\eps>0$, a Steiner $(2+\eps)$-spanner for $K(R,B)$ with $O_\eps(n)$ edges. We may further require that the Steiner spanner is a bipartite graph: For an edge-weighted bipartite graph $G=(R\cup B,E)$ with partite sets $R$ and $B$, a \emph{bichromatic Steiner $t$-spanner} is a Steiner $t$-spanner $H=(\widehat{S},\widehat{E})$ such that $H$ is bipartite with partite sets $\widehat{S}=\widehat{R}\cup \widehat{B}$, where $R\subseteq \widehat{R}$ and  $B\subseteq \widehat{B}$. 
\begin{restatable}{theorem}{SteinerUB}
\label{thm:Steiner}
For every $\eps\in (0,1)$ and every 2-colored set of $n$ points in the plane, one can construct a bichromatic Steiner $(2+\eps)$-spanner with $O(\sqrt{1/\eps}\cdot n)$ edges in $O(\sqrt{1/\eps}\cdot n\log n)$ time. 
\end{restatable}

\paragraph{Bichromatic Spanners on the Line.}
We also study spanners for $K(R,B)$ in the real line. The lower bound construction by Bose et al.~\cite{BoseCCMMS09} can be realized already in the line: It shows that, in general, a bichromatic spanner construction cannot guarantee stretch $t<3$ with $o(n^2)$ edges. We construct a bichromatic 3-spanner with at most $2n-3$ edges (\Cref{prop:3}). We also show that the bichromatic Euclidean MST (i.e., the MST of $K(R,B)$ with Euclidean edge weights) is a 7-spanner; and the stretch bound 7 is the best possible (\Cref{prop:1D}). For comparison, recall that for $n$ (uncolored) points on the line, the Euclidean MST is an $x$-monotone path, which is a 1-spanner. However, for $n$ (uncolored) points in the plane, the Euclidean MST has unbounded stretch. 

\subsection{Related Previous Work}
Bose et al.~\cite{BoseCCMSZ09} studied a related problem: Given a finite set $S$ in the plane, partition $S$ into two nonempty color classes, $S=R\cup B$, so as to minimize the spanning ratio of the resulting complete bipartite graph $G=K(R,B)$, where the weight of every edge equals its Euclidean length. This problem is fundamentally different: In our problem, the color classes and $K(R,B)$ are fixed, and we need to find a small spanner in $K(R,B)$. 

Bereg et al.~\cite{BeregHKTTZ25} recently considered Euclidean spanners for the graph $G(R,B)$ of all red-blue and blue-blue edges (excluding red-red edges) on a set $S=R\cup B$ of $n$ bichromatic points in $\mathbb{R}^d$.  They constructed a 1.998-spanner of size at most $13n$ for $G(R,B)$ in the plane, and showed that the stretch of any spanner for $G(R,B)$ with $O(n)$ edges is at least $1+\eta$ for some absolute constant $\eta>0$. They also constructed a 1-spanner of size at most $2n$ for points in the real line. 

Bichromatic variants for network optimization problems on $n$ points in the plane, restricted to $K(R,B)$, have been studied extensively. In general, these problems are known to be more challenging than their original uncolored versions: For example, the Euclidean MST is always a noncrossing straight-line graph; however, the bichromatic Euclidean MST may have crossing edges~\cite{AgarwalES91,BiniazBMS18}. Researchers studied
bichromatic MSTs~\cite{BiniazBEMMS18}, bichromatic noncrossing MSTs~\cite{AkitayaBDKST25,BorgeltKLLMSV09}, bichromatic bottleneck MSTs~\cite{Abu-AffashBCM21}, and bichromatic bottlneck noncrossing matchings~\cite{BiniazMS14a}. 
Bichromatic TSP~\cite{KarpR14,Nederlof20}, Euclidean TSP~\cite{BaltzS05} 
and bichromatic Voronoi diagrams \cite{BOP25,MantasPSW25} were also studied. 
Many of these problems are already nontrivial for bichromatic points on a line~\cite{Abu-AffashBCC19,BandyapadhyayBB21}.

\subsection{Technical Highlights and Comparison to Previous Techniques} 
As mentioned above, Bose et al.~\cite{BoseCCMMS09} constructed a $(3+\eps)$-spanner with $O(n\log n)$ edges for a bichromatic set $S=R\cup B$ of $n$ points and constant $\eps>0$. Their spanner construction uses a hierarchical space partition, but a crucial step of their approach has a $\Omega(n \log n)$ barrier. Our approach is significantly different and perhaps conceptually simpler. We briefly review the spanner construction in~\cite{BoseCCMMS09} and point out key differences from our construction.

For a bichromatic set $S=R\cup B$ of $n$ points in the plane, Bose et al.~\cite{BoseCCMMS09} use the classical \emph{well-separated pair decomposition} (WSPD, for short) by Callahan and Kosaraju~\cite{CallahanK95}. A WSPD with separation parameter $\varrho>0$ is a collection of unordered pairs $\mathcal{P}=\big\{\{S_i,T_i\}: i=1,\ldots ,m\big\}$ such that for every point pair $\{s,t\}\in \binom{S}{2}$ there is an $i$ such that $|\{s,t\}\cap S_i|=1$ and $|\{s,t\}\cap T_i|=1$; and $\dist(S_i,T_i)\geq \varrho\cdot \max\{\diam(S_i),\diam(T_i)\}$ for all $i$. A WSPD of size $m=O(n)$ can be constructed using split trees~\cite{CallahanK95,NS07} or compressed quadtrees~\cite{HarPeled11}. For an uncolored set of $n$ points in $\mathbb{R}^d$, WSPDs with separation parameter $\varrho=\Theta(1/\eps)$ were previously used for constructing $(1+\eps)$-spanners of size $O(n/\eps^d)$~\cite{CallahanK95,FischerH05,HarPeledM06}; and this bound is the best possible. 
Given a WSPD, each set $S_i$ (resp., $T_i$) could be either monochromatic or bichromatic. 
Bose et al.~\cite{BoseCCMMS09} choose an arbitrary \emph{representative} from each monochromatic set $S_i$ and $T_i$; and two representatives of different colors from bichromatic sets. 
Their spanner consists of specific edges chosen for each pair $\{S_i,T_i\}$ of the WSPD. If $S_i$ and $T_i$ are bichromatic, they add a single edge between a pair of representatives from $S_i$ and $T_i$, respectively. However, if a set $S_i$ is monochromatic, they need a path from all points of $S_i$ to its representative: They achieve this by connecting all points in $S_i$ to the representative of a closest cluster $S_j=\mathrm{cl}(S_i)$ that contains a point of the other color (hence every point in $S_i$ will have a 2-hop path to the representative of $S_i$); refer to \Cref{fig:bose2009}. They prove that this graph is a $(5+\eps)$-spanner of size $O(n)$. They further show that \emph{if} the WSPD has the additional property that 
$\min\{|S_i|,|T_i|\}=1$ for all $i$, then their construction yields a $(3+\eps)$-spanner. However, 
a WSPD with that additional property has size $m=O(n\log n)$, which is the best possible~\cite{CallahanK95,NS07}, and dominates the size of their $(3+\eps)$-spanner.
(We note that  Abam et al.~\cite[Section~4.1]{AbamCFS13} later proposed a randomized algorithm for constructing a bichromatic $(3+\eps)$-spanner of size $O(n \log n)$ using \emph{semi-separated pair decompositions} (SSPD) in which each point pair appears in $O(\log n)$ semi-separated pairs with high probability~\cite{AbamH12}.)

    \begin{figure}[tbh]
        \centering
        \includegraphics[width=.65\textwidth]{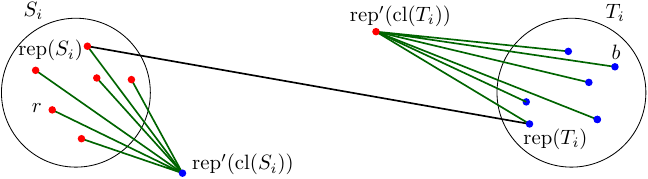}
        \caption{A well-separated pair $\{S_i,T_i\}$ of a red cluster $S_i\subset R$ and a blue cluster $T_i\subset B$, with representatives $\mathrm{rep}(S_i)\in S_i$ and $\mathrm{rep}(T_i)\in T_i$, resp.,  and representatives $\mathrm{rep}'(\mathrm{cl}(S_i))\in B$ and $\mathrm{rep}'(\mathrm{cl}(T_i))\in R$ of the opposite colors in closest clusters. For a red-blue point pair $(r,b)\in S_i\times T_i$, a 5-edge $rb$-path is $(r,\mathrm{rep}'(\mathrm{cl}(S_i)),\mathrm{rep}(S_i), \mathrm{rep}(T_i), \mathrm{rep}'(\mathrm{cl}(T_i)), b)$.}
        \label{fig:bose2009}
    \end{figure}

Our spanner design is different in the crucial step that handles well-separated pairs of  monochromatic clusters. If both $S_i$ and $T_i$ are bichromatic, we add just two edges between representatives, similarly to \cite{BoseCCMMS09}. 
If $S_i$ is monochromatic, however, we connect all points in $S_i$ to the representative of $T_i$ of the opposite color (and similarly, if $T_i$ is monochromatic, we connect all points in $T_i$ to a representative of $S_i$); see \Cref{fig:new}. Overall, for each well-separated pair $(S_i,T_i)$, we add a star, a bistar, or just a pair of edges between $S_i$ and $T_i$ (see Cases 1--3 in \Cref{subsec:construct}). In particular, we do not use a detour via a closest cluster $S_j=\mathrm{cl}(S_i)$. This new approach would already give a bichromatic $(3+\eps)$-spanner of size $O(n)$ for any constant $\eps>0$.
    
    \begin{figure}[tbh]
        \centering
        \includegraphics[width=.65\textwidth]{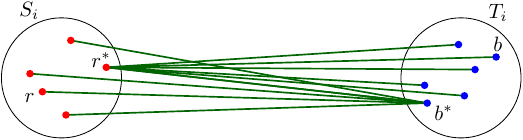}
        \caption{A well-separated pair $\{S_i,T_i\}$ of a red cluster $S_i\subset R$ and a blue cluster $T_i\subset B$, with representatives $r^*\in S_i$ and $b^*\in T_i$, respectively. For a red-blue point pair $(r,b)\in S_i\times T_i$, a 3-edge $rb$-path is $(r,b^*,r^*,  b)$.}
        \label{fig:new}
    \end{figure}
    
Instead of a quadtree-based WSPD, we use a custom-made space partition scheme with rectangular cells of aspect ratio $\Theta(\sqrt{1/\eps})$ to improve the dependence on $\eps>0$. In a standard quadtree-based WSPD with separation parameter $\varrho=\Theta(1/\eps)$, the number of pairs is $m=O(\varrho^2 n)=O(n/\eps^2)$~\cite[Theorem~4.1.3]{HarPeled11}. 
Even though the separation parameter must be $\varrho=\Theta(1/\eps)$ for a $(1+\eps)$-spanner on uncolored point sets, our stretch analysis for a $(3+\eps)$-spanner also works with a smaller separation parameter $\varrho=\Theta(\sqrt{1/\eps})$.
This would reduce the size of our spanner to $O(n/\eps)$. 

Furthermore, our inductive stretch analysis goes through even if we replace the square cells of a quadtree with rectangular cells with a larger aspect ratio provided that any two edges between  the well-separated cells make an angle $O(\sqrt{\eps})$. This relaxes the separation condition to $\dist(S_i,T_i)\geq \Theta(\sqrt{1/\eps}) \cdot \max\{\mathrm{width}(S_i),\mathrm{width}(T_i)\}$, and further reduces the size of our $(3+\eps)$-spanner to $O(n/\sqrt{\eps})$. We present this construction in \Cref{subsec:tiling,subsec:construct}. The rectangular shape of the cells matters for the stretch analysis in \Cref{subsec:analysis}, however, the sparsity analysis still relies on  compressed quadtrees (after suitable affine transformations) in  \Cref{subsec:analysis2}.  

\paragraph{Dependence on $\eps>0$.}
In \Cref{sec:lower}, we show that the dependence on $\eps>0$ is necessary.
Recall that for $n$ (uncolored) points in Euclidean plane, there exist $(1+\eps)$-spanners with $O(n/\eps)$ edges, and this bound is the best possible: The lower bound construction is a set of $\Theta(1/\eps)$ equally spaced points along the boundary of a disk or square. We use a similar construction for colored points, and place $\Theta(\sqrt{1/\eps})$ equally spaced red and blue points on two opposite sides of a unit square. However, the analysis of a bichromatic $(3+\eps)$-spanner is significantly more involved than for uncolored points. Stretch $1+\eps$ leaves very small room for error: An $ab$-path of length $(1+\eps)|ab|$ must lie in $O(\sqrt{\eps})|ab|$-neighborhood of the line segment $ab$. With stretch $3+\eps$, we have to analyze $rb$-paths that venture further from $r$ and $b$. Based on the location of the points (on opposite sides of the square), we show that every $rb$-path of length at most $(3+\eps)|rb|$ has at most 3 edges (\Cref{lem:onethree}), and lies in a axis-aligned box whose height is at most $7\, |y(r)-y(b)|$  if $\sqrt{\eps}\leq |y(r)-y(b)|\leq \frac13$ (\Cref{lem:y9h}). These properties are sufficient to show that any $(3+\eps)$-spanner must have $\Omega(n)$ edges at each scale, for scales that increase doubly exponentially from $\sqrt{\eps}$ to $\frac13$, hence $\Omega(\log \log \eps^{-1})$ different scales.  

\paragraph{Off-Metric Steiner Points.}
Recall that, without Steiner points, stretch $3-\eps$ is impossible with $o(n^2)$ edges because a pair of well-separated red and blue clusters $(R,B)$ requires a complete bipartite graph. However, for stretch $3+\eps$, a bistar with $O(n)$ edges is sufficient for two such clusters: with star centers $r^*\in R$ and $b^*\in B$. For any edges $rb$ for $r\in R\setminus \{r^*\}$ and $b\in B\setminus \{b^*\}$, an $rb$-path $P_{rb}=(r,b^*,r^*,b)$ travels between the two clusters three times. 
If off-metric Steiner points are allowed, we can create Steiner star centers $r^*$ and $b^*$ such that $w(r^*,b^*)\approx 0$, which reduces the length of $P_{rb}$ to $(2+\eps)|rb|$. It turns out that our non-Steiner construction can be easily adapted to this setting, using  Steiner points only between pairs of  monochromatic clusters/tiles: We can construct a bichromatic Steiner $(2+\eps)$-spanner with $O(n/\sqrt\eps)$ edges
(\Cref{ssec:SteinerUB}). 
We also show that, even with off-metric Steiner points, $K(R,B)$ does not always admit a Steiner $(2-\eps)$-spanner with $o(n^2)$ edges for any constant $\eps>0$ (\Cref{ssec:SteinerLB}). 
The lower bound construction consists of a well-separated pair of a red cluster $R$
and a blue cluster $B$. We show that if $H$ is a bichromatic Steiner $(2-\eps)$-spanner, then for every red-blue pair $(r,b)\in R\times B$ there is a unique edge $e_{r,b}$ on a shortest $rb$-path, or else we would find two points of the same color, say $r_1,r_2\in R$, such that $d_H(r_1,r_2)< d_{K(R,B)}(r_1,r_2)$, contradicting the noncontracting condition.  


\section{Bichromatic $(3+\eps)$-Spanners in the Plane}
\label{sec:2D}

In this section, we construct a bichromatic $(3+\eps)$-spanner for a 2-colored point set in the plane in \Cref{sec:2D}. A system of tilings with  exponentially increasing scales and $\Theta(\sqrt{\eps})$ directions is presented in \Cref{subsec:tiling}. The spanner is constructed based on pairs of ``well-separated'' tiles in \Cref{subsec:construct}. The stretch and sparsity analyses are presented in \Cref{subsec:analysis} and \Cref{subsec:analysis2}, respectively. We analyze the running time of our spanner construction in~\Cref{ssec:runtime}, and briefly outline generalizations to dimensions $d\geq 3$ in \Cref{ssec:dspace}.

\subsection{Tilings of the Plane}
\label{subsec:tiling} 

We begin by establishing some preliminary definitions. For a given $\eps \in (0,1)$ and a separation constant $\mu \in \{1,2,\ldots,14\}$, let $\delta = \sqrt{ \eps /7} = \Theta(\sqrt{\eps})$. We construct a tiling of the plane with rectangular tiles of aspect ratio $1/\delta$. 
\begin{definition} \label{def:tiling}
   For every $i,j\in \mathbb{Z}$, we define the \emph{tile}
        \[ 
        T(i,j)= [i-1,i) \times [(j-1)\delta, j\delta).
        \]
    The set of all tiles forms a \emph{tiling} $T=\{T(i,j): i,j \in \mathbb{Z}\}$.
\end{definition}

In our spanner construction, we use several affine copies of the tiling $T$, each obtained by transforming all tiles in $T$ by the same affine transformation (rotation, scaling, and translation). Recall that every tile $T(i,j)\in T$ is an axis-aligned rectangle of width 1. Let $T'=\{T(i,j)+(0,\frac\delta2): i,j\in \mathbb{Z}\}$ denote the tiling obtained by $T$ by translation, with vector $(0,\frac\delta2)$. For every angle $\theta\in [0,2\pi)$ and every $\lambda>0$, let $T_{\theta,\lambda}$ (resp., $T'_{\theta,\lambda}$) be the tiling obtained from $T$ (resp., $T'$) by a counterclockwise rotation through $\theta$ about the origin, followed by a dilation with factor $\lambda$. 

However, we restrict ourselves to a discrete family of tilings: We use only $O(\sqrt{1/\eps})$ angles $\theta$, and a discrete set of exponentially increasing scaling factors, that we define now. Let $\varphi = \arctan(\delta / (2\mu+4))$  and $k = \lceil \pi / \varphi \rceil$. Note that $\varphi\in \Theta (\sqrt{\eps})$ and $k\in \Theta(\sqrt{1/\eps})$.
Our discrete set of $k$ angles is $\Phi =  \{2 r \pi / k : r = 0,1,\ldots , k-1\}$. 
Let $\gamma =3^{1/(\mu+2)}$, and note that $3^{1/(\mu+2)}< 1+1/\mu$.
Our discrete set of scaling factors is $\Lambda = \{ \gamma^p : p \in \mathbb{Z}\}$. 
The family of all tilings with rotation angle in $\Phi$ and scaling factor in $\Lambda$ is  
    \[\mathcal{T} = \bigcup_{\theta \in \Phi}\bigcup_{\lambda \in \Lambda}  \big( T_{\theta, \lambda}\cup T_{\theta, \lambda}'\big).\] 

The following lemma provides an important property of the family of tilings $\mathcal{T}$, which also plays a crucial role in the stretch analysis of our construction.

\begin{lemma}\label{lemma:spacing}
    For every pair of distinct points $a,b \in \mathbb{R}^2$, there exist $\theta^* \in \Phi$,  $\lambda^* \in \Lambda$ and $i^*,j^* \in \mathbb{Z}$ such that at least one of the following holds:
    \begin{enumerate}
        \item tiles $T_{\theta^*, \lambda^*}(i^*,j^*)$ and $T_{\theta^*, \lambda^*}(i^*+\mu+1,j^*)$ each contain one of $a$ and $b$,
        \item tiles $T_{\theta^*, \lambda^*}'(i^*,j^*)$ and $T_{\theta^*, \lambda^*}'(i^*+\mu+1,j^*)$ each contain one of $a$ and $b$.
    \end{enumerate}
\end{lemma}

\begin{proof}
    Let $\psi$ be the counterclockwise angle between the positive $x$-axis and vector $\vec{ab}$. We choose $\theta^* \in \Phi$ such that $\theta^*$ is the closest angle to $\psi$; i.e., $\theta^*$ is the solution to the minimization problem
    \[\theta^* = \argmin_{\theta \in \Phi} \lvert \theta - \psi \rvert.\]
    We may assume (by suitable rotation and reflections, if necessary) that $\theta^* = 0$, $\varphi\geq 0$, and $x(a) < x(b)$, where $x(a)$ and $x(b)$, resp., denote the $x$-coordinate of $a$ and $b$.
    Next, we show that there exists a scale $\lambda^*\in \Lambda$ that satisfies 
    \begin{equation}\label{eq:scale}
    \mu \lambda^* \le x(b)-x(a) < (\mu+1)\lambda^*.
    \end{equation}
    Let $\lambda^* = \gamma^{p^*}$, where 
    \[
    p^* = \left\lfloor \frac{\ln\frac{ x(b)-x(a)}{\mu}}{\ln\gamma} \right\rfloor .
    \]
   This gives $\gamma^{p^*}\leq \frac{ x(b)-x(a)}{\mu}\leq \gamma^{p^*+1}$; and combined with $\gamma=3^{1/(\mu+2)}< 1 + 1/\mu$, it implies \Cref{eq:scale}.
    
    Since $T_{\theta^*, \lambda^*}$ is a tiling of $\mathbb{R}^2$, there exist $i^*,j^*\in \mathbb{Z}$ such that $a \in T_{\theta^*, \lambda^*}(i^*,j^*)$. Now \Cref{eq:scale}
    implies that $x(a)+\mu \lambda^*\leq x(b)\leq x(a)+(\mu + 1)\lambda^*$, and so $b \in T_{\theta^*, \lambda^*}(i^*+\mu +1,j^{**})$ for some $j^{**} \in \mathbb{Z}$.
    Let $y(a)$ and $y(b)$, resp., denote the $y$-coordinate of $a$ and $b$. If $y(a) = y(b)$, then we must have $j^*= j^{**}$. Suppose $y(a) \ne y(b)$, we may assume w.l.o.g.\ that $y(a) < y(b)$ and $\lambda^* = 1$. Under these assumptions, we have $T_{\theta^*, \lambda^*}(i^*,j^*)=T(i^*,j^*)$, which is an axis-aligned rectangle of width 1 and height $\delta$. 
    By our choice of $\varphi$, we can bound $y(b)$ from above, and conclude that 
    \begin{equation}\label{eq:y}
    y(b) - y(a) = (x(b) - x(a)) \tan(\psi) 
    \le (\mu + 2)\tan(\varphi) 
    =  \frac{\delta}{2} .
    \end{equation}
    A horizontal median of the tile $T_{\theta^*, \lambda^*}(i^*,j^*)$ subdivides $T_{\theta^*, \lambda^*}(i^*,j^*)$ into two congruent rectangles, that we call the \emph{upper} and \emph{lower half} of $T_{\theta^*, \lambda^*}(i^*,j^*)$. Then one of the following two cases occurs
    \begin{enumerate}
        \item[Case~1:] $a$ is in the lower half of $T_{\theta^*, \lambda^*}(i^*,j^*)$. In this case, we have $0\leq y(a)\leq \delta/2$, and \Cref{eq:y} gives $0\leq y(b)\leq \delta/2$. Thus $a \in T_{\theta^*, \lambda^*}(i^*,j^*)$ and $b \in T_{\theta^*, \lambda^*}(i^*+\mu + 1,j^*)$.        
        \item[Case~2:] $a$ is in the upper half of $T_{\theta^*, \lambda^*}(i^*,j^*)$. In this case, we have $\delta/2 \leq y(a)\leq \delta$, and \Cref{eq:y} gives $\delta/2 \leq y(b)\leq \frac32 \delta$. Consequently, $a \in T_{\theta^*, \lambda^*}'(i^*,j^*)$ and $b \in T_{\theta^*, \lambda^*}'(i^*+\mu+1,j^*)$.
    \end{enumerate}
    This concludes the proof of \Cref{lemma:spacing}.
\end{proof}

\subsection{Spanner Construction}
\label{subsec:construct}
We are given a set $S = R \cup B$ of $n$ bichromatic points in the plane and a parameter $\eps > 0$.
Let $\mathcal{T}$ be the family of tilings defined in \Cref{subsec:tiling} with the separation constant $\mu = 9$. 
For every nonempty tile $T_{\theta,\lambda}(i,j)$, we denote by $N_{\theta,\lambda}(i,j)$ its \emph{neighborhood} obtained by taking the union of $T_{\theta,\lambda}(i,j)$ and its eight surrounding tiles at the same scale. We say that $N_{\theta,\lambda}(i,j)$ is \emph{bichromatic} if it contains at least one point of each color. Similarly, $N_{\theta,\lambda}(i,j)$ is \emph{monochromatic} if all points in $N_{\theta,\lambda}(i,j)$ have the same color. It is important to note that $N_{\theta,\lambda}(i,j)$ can be bichromatic even if all points in $T_{\theta,\lambda}(i,j)$ are of the same color.
For our construction, we begin by initializing a graph $G = (S,E)$ to be the empty graph. 

For a set $Q\subseteq S$, we say that point $q$ is \emph{leftmost} in $Q\cap N_{\theta,\lambda}(i,j)$ if, after a clockwise rotation of $Q\cap N_{\theta,\lambda}(i,j)$ through $\theta$, its $x$-coordinate is minimal. Likewise, the point $q$ is \emph{rightmost} in $Q\cap N_{\theta,\lambda}(i,j)$ if, after the same rotation, its $x$-coordinate is maximal. Note that several leftmost (resp.,  rightmost) points may exist (which have the same $x$-coordinate after suitable rotation). 
We choose a rightmost red point $r^*$ in $N_{\theta, \lambda}(i,j)$ and a leftmost blue point $b^*$ in $N_{\theta,\lambda}(i+10,j)$.
Similarly, we choose a rightmost blue point $b^{**}$ in $N_{\theta,\lambda}(i,j)$ and a leftmost red point $r^{**}$ in $N_{\theta,\lambda}(i+10,j)$.

Consider every $\theta \in \Phi$, $\lambda \in \Lambda$, and $i,j\in \mathbb{Z}$ such that both $T_{\theta,\lambda}(i,j)$ and $T_{\theta,\lambda}(i+10,j)$ are nonempty
(the translated tiles are handled similarly if both $T'_{\theta,\lambda}(i,j)$ and $T'_{\theta,\lambda}(i+10,j)$ are nonempty).
We distinguish between four cases based on whether the neighborhoods $N_{\theta,\lambda}(i,j)$ and $N_{\theta,\lambda}(i+10, j)$ are monochromatic or bichromatic.

\begin{enumerate}
    \item[Case~0:] Both $N_{\theta,\lambda}(i,j)$ and $N_{\theta,\lambda}(i+10, j)$ are monochromatic, and all points in $N_{\theta,\lambda}(i,j)\cup N_{\theta,\lambda}(i+10,j)$ have the same color. In this case, we do not add any new edges to $G$.
    \item[Case~1:] Both $N_{\theta,\lambda}(i,j)$ and $N_{\theta,\lambda}(i+10,j)$ are monochromatic, but they contain different colors. Assume w.l.o.g.\ that $N_{\theta,\lambda}(i,j)\cap S\subset R$ and $N_{\theta,\lambda}(i+10,j)\cap S\subset B$  (the other case is symmetric); see \Cref{fig:monomono}. In this case, we add edges to $G$ between $r^*$ and every point in $T_{\theta,\lambda}(i+10,j)\cap B$, between $b^*$ and every point in $T_{\theta,\lambda}(i,j)\cap R$, as well as the edge $r^*b^*$. In particular, the edges added to $G$ form a \emph{bistar} centered at $r^*$ and $b^*$.
    \begin{figure}[tbh]
        \centering
        \includegraphics[width=.9\textwidth]{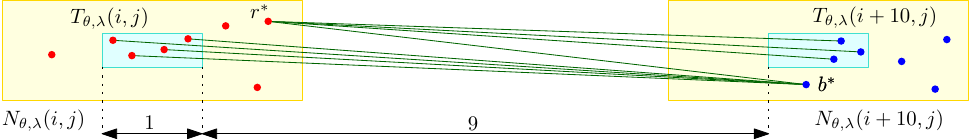}
        \caption{Case~1: Both $N_{\theta,\lambda}(i,j)$ and $N_{\theta,\lambda}(i+10, j)$ are monochromatic}
        \label{fig:monomono}
    \end{figure}
    \item[Case~2:] Exactly one of $N_{\theta,\lambda}(i,j)$ and $N_{\theta,\lambda}(i+10,j)$ is monochromatic. Assume w.l.o.g.\ that $N_{\theta,\lambda}(i,j)$ is monochromatic and $N_{\theta,\lambda}(i+10,j)$ is bichromatic; see \Cref{fig:monobi} (the other case is symmetric). If $N_{\theta,\lambda}(i,j)\cap S\subset R$, we add an edge to $G$ between $b^*$ and every point in $T_{\theta,\lambda}(i, j)\cap R$. Otherwise,  (when $N_{\theta,\lambda}(i,j)\cap S\subset B$), we add an edge to $G$ between $r^{**}$ and every point in $T_{\theta,\lambda}(i, j)\cap B$. In this case, the edges added to $G$ form a \emph{star}.
    \begin{figure}[tbh]
        \centering
        \includegraphics[width=.9\textwidth]{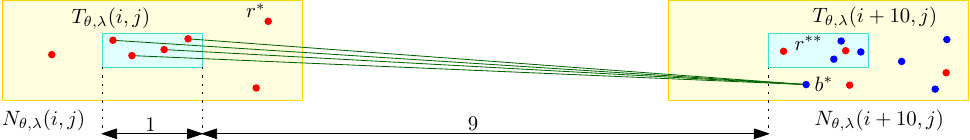}
        \caption{Case~2: $N_{\theta,\lambda}(i,j)$ is monochromatic and $N_{\theta,\lambda}(i+10, j)$ is bichromatic}
        \label{fig:monobi}
    \end{figure}
    \item[Case~3:] Both $N_{\theta,\lambda}(i,j)$ and $N_{\theta,\lambda}(i+10,j)$ are bichromatic; see  \Cref{fig:bibi}. In this case, we add two edges to $G$: the edges $r^*b^*$ and $r^{**}b^{**}$.
    \begin{figure}[tbh]
        \centering
        \includegraphics[width=.9\textwidth]{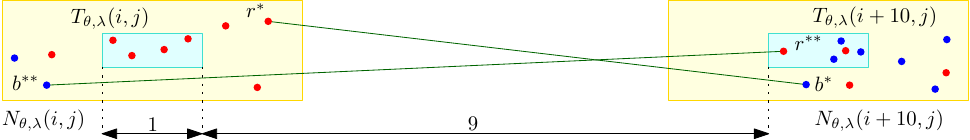}
        \caption{Case~3: Both $N_{\theta,\lambda}(i,j)$ and $N_{\theta,\lambda}(i+10, j)$ are bichromatic}
        \label{fig:bibi}
    \end{figure}
\end{enumerate}

This completes the construction of the bipartite graph $G=(R\cup B,E)$. We show that $G$ is a $(3+\eps)$-spanner of $K(R,B)$ in \Cref{subsec:analysis}, and that it contains $O(n/\sqrt{\eps})$ edges in \Cref{subsec:analysis2}.
We describe how to compute $G$ in $O(\sqrt{1/ \eps}\cdot n\log n)$ time in \Cref{ssec:runtime}.

\subsection{Stretch Analysis}
\label{subsec:analysis}

In this section, we prove that the graph $G$ is a $(3+\eps)$-spanner of $K(R,B)$. We begin with a preliminary observation. The Euclidean length of an edge is given by the Pythagorean Theorem for its $x$- and $y$-extents. However, if the $y$-extent is much smaller than the $x$-extent, we can bound the Euclidean length from above by a polynomial as follows: For all $x \ge 1$ and $y > 0$, we have 
\begin{equation}\label{eq:P}
\sqrt{x^2 + y^2} \le x + y^2/2.
\end{equation}
We can now prove the main result of this subsection. 
\begin{lemma}\label{lem:stretch}
For every finite bichromatic set $S=R\cup B$ in the plane, the graph $G\subseteq K(R,B)$ constructed in \Cref{subsec:construct} is a $(3+\eps)$-spanner for $K(R,B)$.
\end{lemma}
\begin{proof}
    We proceed by induction on the length of the bichromatic edges in $K(R,B)$. 
    Specifically, let $D=\{ |rb|: (r,b)\in R\times B\}$ be the set of distances between red-blue pairs in $S$; and let $0=d_0<d_1<\ldots <d_m$ be elements of $\{0\}\cup D$ in increasing order. 
    
    We prove, by induction on $\ell=0,1,\ldots , m$, that for all $(r,b)\in R\times B$ with $|rb|\leq d_\ell$, the graph $G$ contains an $rb$-path of length at most $(3+\eps)|rb|$. In the base case, where $\ell=0$ and $d_\ell=0$, the claim vacuously holds, as there is no red-blue pair $(r,b)$ with $|rb|=0$. 

    In the induction step, let $d_\ell\in D$ and assume that for all $(r',b')\in R\times B$, if $|r'b'|<d_\ell$, then $G$ contains an $r'b'$-path of length at most $(3+\eps)|r'b'|$.  
    Let $(r,b)\in R\times B$ with $|rb|=d_\ell$. By Lemma~\ref{lemma:spacing}, there exist $\theta\in \Phi$, $\lambda \in \Lambda$ and $i,j \in \mathbb{Z}$ such that $T_{\theta,\lambda}(i,j)$ and $T_{\theta,\lambda}(i+10,j)$ each contain one of $r$ and $b$; or $T_{\theta,\lambda}'(i,j)$ and $T_{\theta,\lambda}'(i+10,j)$ each contain one of $r$ and $b$.  To simplify our argument, we may assume w.l.o.g. that $\theta = 0$, $\lambda = 1$, $r\in T_{\theta,\lambda}(i,j)$ and $b \in T_{\theta,\lambda}(i+10,j)$. Under these assumptions, $N_{\theta, \lambda}(i,j) = N(i,j)$, which is an axis-aligned rectangle of width $3$ and height $3\delta$.  We distinguish between three cases:
    \begin{enumerate}
        \item [Case~1:] Both $N(i,j)$ and $N(i+10,j)$ are monochromatic; see \Cref{fig:monomono}.
        By construction (Case 1), $G$ contains the $rb$-path $P = (r,b^*,r^*,b)$. Since the pair $r^*$ is the rightmost red point in $N(i,j)$ and $b^*$ is the leftmost blue point in $N(i+10,j)$, then we have $x(r) \le x(r^*) < x(b^*) \le x(b)$, where $x(r)$, $x(r^*)$, $x(b^*)$ and $x(b)$, resp., denote the $x$-coordinate of $r$, $r^*$, $b^*$ and $b$.
        Then, the length of $b^*r^*$ can be bounded above by $|b^*r^*| \le \sqrt{(x(b^*) - x(r^*))^2 + (3\delta)^2} \le \sqrt{(x(b) - x(r))^2 + (3\delta)^2}$. Using \Cref{eq:P}
        and $9 \le |rb|$, it follows that 
        \begin{align}\label{eq:case1}
            |b^*r^*| &\le \sqrt{(x(b) - x(r))^2 + (3\delta)^2} \\
            &\le x(b) - x(r) + 9\delta^2 / 2 \nonumber\\
            &= x(b) - x(r) + 9\eps / 14  \nonumber\\
            &\le (1 + \eps / 3)|rb|. \nonumber
        \end{align}
        By a similar process, we obtain $|rb^*| \le (1 + \eps/3)|rb|$ and $|r^*b| \le (1 + \eps/3)|rb|$. Thus, the length of the $rb$-path $P$ is at most $(3+\eps)|rb|$ since $|P| = |rb^*| + |b^*r^*| + |r^*b| \le 3(1+\eps / 3)|rb| = (3+\eps)|rb|$.
        (Note that we did not use the induction hypothesis in Case~1.)
        \item [Case~2:] One of $N(i,j)$ and $N(i+10,j)$ is monochromatic and the other is bichromatic. Assume w.l.o.g.\ that $N(i,j)$ is monochromatic and $N(i+10,j)$ is bichromatic; see \Cref{fig:stretch1} (the other case is symmetric). Let $\tilde{r} \in R\cap N(i+10,j)$. Since $b^*,\tilde{r}\in N(i+10,j)$, we have $|b^*\tilde{r}|\leq \diam(N(i+10,j))=3\diam(T(i,j))$. Furthermore, since $b\in T(i+10,j)$, then the edge $\tilde{r}b$ lies in an axis-aligned rectangle spanned by $T(i+10,j)$ and an adjacent tile, which implies $|\tilde{r}b|\leq 2\diam(T(i,j))$. Note that $\diam(T(i,j))=\sqrt{1+\delta^2}$ by the Pythagorean Theorem. \Cref{eq:P} yields $\diam (T(i,j)) = \sqrt{1+\delta^2} \le 1 + \eps / 14$.
        Since $9 \le |rb|$, it now follows that $|b^* \tilde{r}| \le 3\diam(T(i,j)) \le 3(1 + \eps / 14) < 9 \leq  |rb|$. Similarly, we have $|\tilde{r}b| < |rb|$.
   \begin{figure}[tbh]
        \centering
        \includegraphics[width=.95\textwidth]{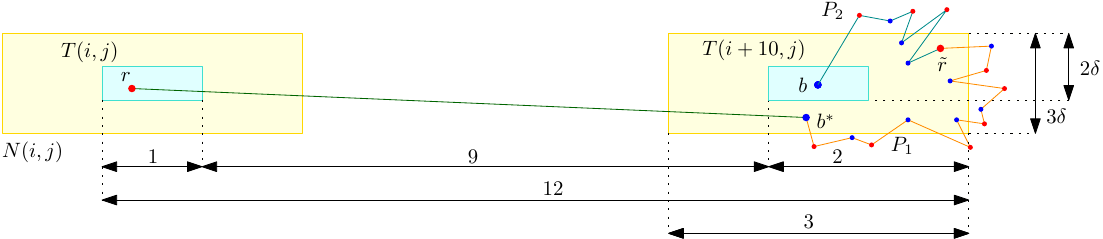}
        \caption{An $rb$-path when $N(i,j)$ is monochromatic and $N(i+10, j)$ is bichromatic}
        \label{fig:stretch1}
    \end{figure}
    
        By the induction hypothesis, 
        $G$ contains some paths
        \[P_1 = (b^*,\ldots,\tilde{r}) \quad\text{and}\quad P_2 = (\tilde{r},\ldots,b),\]
        resp., of length at most $(3+\eps)|b^*\tilde{r}|$ and $(3+\eps)|\tilde{r}b|$. Then, the concatenation of $rb^* \oplus P_1 \oplus P_2$ is an $rb$-path in $G$ and its length can be bounded above as follows
        \begin{align}\label{eq:case2}
            |rb^*| + |P_1| + |P_2| &\le \sqrt{11^2 + (2\delta)^2} + (3+\eps)|b^*\tilde{r}| + (3+\eps)|\tilde{r}b| \\
            &\le 11 + 2\delta^2 + 3(3+\eps)\diam(T(i,j)) + 2(3+\eps)\diam(T(i,j)) \nonumber\\
            &= 11 + 2\eps / 7+ 5(3+\eps)(1+\eps / 14) \nonumber\\
            &= 11 + 2\eps / 7 + 5(3 + 17\eps / 14 + \eps^2 / 14).\nonumber
        \end{align}
        Recall that $0< \eps < 1$, hence $\eps^2 < \eps$. It follows that
        \begin{align*}
            |rb^*| + |P_1| + |P_2| &\le 11 + 2\eps / 7 + 5(3 + 9\eps /7) \\ 
            &= 26 + 47 \eps /7\\
            &< 9(3+\eps) \\
            &\le (3+\eps)|rb|.
        \end{align*}
        So, the concatenation of $rb^* \oplus P_1 \oplus P_2$ is an $rb$-path in $G$ of length at most $(3+\eps)|rb|$.
    \begin{figure}[tbh]
        \centering
        \includegraphics[width=.95\textwidth]{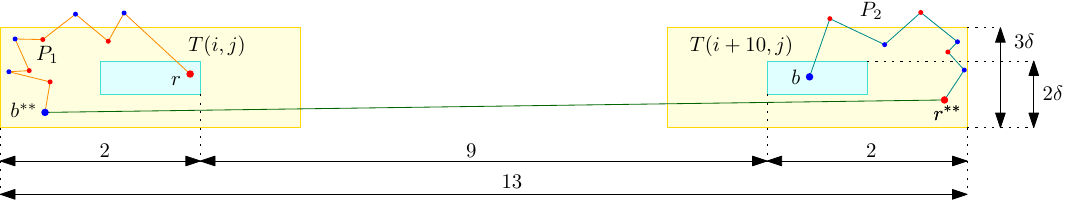}
        \caption{An $rb$-path when both $N(i,j)$ and $N(i+10, j)$ are bichromatic}
        \label{fig:stretch2}
    \end{figure}
        \item [Case~3:] Both $N(i,j)$ and $N(i+10,j)$ are bichromatic; see \Cref{fig:stretch2}. 
        Similarly to Case~2, we have $|rb^{**}|\leq \diam N(i,j)< |rb|$ and  $|r^{**}b|\leq \diam N(i+10,j)< |rb|$. By the induction hypothesis, $G$ contains some paths
        \[P_1 = (r,\ldots,b^{**}) \quad\text{and}\quad P_2 = (r^{**},\ldots,b),\]
        resp., of length at most $(3+\eps)|rb^{**}|$ and $(3+\eps)|r^{**}b|$. Then, the concatenation $P_1 \oplus b^{**}r^{**} \oplus P_2$ is an $rb$-path of $G$, and observe that by a similar process as in Case~2, its length can be bounded above as follows:
        \begin{align}\label{eq:case3}
            |P_1| + |b^{**}r^{**}| + |P_2| \
            &\le (3+\eps)|rb^{**}| + \diam(N(i,j) \cup N(i+10,j))+ (3+\eps)|r^{**}b|\\
            &\le 2(3+\eps) \diam(T(i,j)) + \sqrt{13^2 + (3\delta)^2} + 2(3+\eps) \diam(T(i+10,j)) \nonumber\\
            &\le 4(3+\eps)(1+ \eps / 14) + 13 + 9\delta^2 / 2\nonumber\\
            &\le 25 + 81 \eps / 14\nonumber\\
            &< 9(3+\eps) \nonumber\\
            &\le (3+\eps)|rb|.\nonumber
        \end{align}
        So, the concatenation $P_1 \oplus b^{**}r^{**} \oplus P_2$ is an $rb$-path in $G$ of length at most $(3+\eps)|rb|$.
    \end{enumerate}
    Overall, the bipartite graph $G$ is a $(3+\eps)$-spanner of $K(R,B)$.    
\end{proof}

\subsection{Sparsity Analysis}
\label{subsec:analysis2}

Recall that we used tilings with separation constant $\mu=9$ for constructing a bichromatic spanner $G$. We analyze the construction for separation constants $9\leq \mu \leq 12$ because then the same sparsity analysis applies other variants of the problem (multichromatic version in \Cref{ssec:multichromatic} and Steiner version in \Cref{sec:Steiner}). Fix an angle $\theta\in \Phi$. We distinguish between two types of edges in $G$.
An edge is of \emph{Type~1} if it was added when the construction processed  a pair of tiles $T_{\theta,\lambda}(i,j), T_{\theta,\lambda}(i+\mu+1,j)$ where one of the tiles had monochromatic neighborhood; otherwise it is of \emph{Type~2}. In particular, edges of Type~1 correspond to Cases~1-2, and Type~2 to Case~3 of the construction.

\paragraph{Overview.}
We count edges of each type separately in \Cref{lem:type1,lem:type2} below. 
Edges of Type~1 are the main concern, since our spanner may contain many edges of Type~1 for each pair 
$T_{\theta,\lambda}(i,j), T_{\theta,\lambda}(i+\mu+1,j)$. However, even if the neighborhood of one of the tiles is monochromatic, their union $N_{\theta,\lambda}(i,j)\cup N_{\theta,\lambda}(i+\mu+1,j)$ contains both colors. We show (\Cref{lem:type1-pre}) that if we go up $26$ scales higher, then every point in $T_{\theta,\lambda}(i,j)\cup T_{\theta,\lambda}(i+\mu+1,j)$ is contained in a bichromatic neighborhood: It then follows that each point $s\in S$ is incident to an edge of type~1 in $O(1)$ scales; hence there are $O(n)$ edges of Type~1 (for each $\theta\in \Phi$). 

For every pair $T_{\theta,\lambda}(i,j), T_{\theta,\lambda}(i+\mu+1,j)$ with bichromatic neighborhoods, we add only two edges of Type~2, so it is enough to analyze the number of such pairs. We rely on the classic result that a quadtree-based WSPD has $m=O(n)$ well-separated pairs~\cite[Theorem~4.1.4]{HarPeled11}. The proof of this result hinges on properties of quadtrees and compressed quadtrees: The compressed quadtree for $n$ points  is known to have $O(n)$ nodes, but chains of single-child nodes in the uncompressed quadtree must be handled carefully. We adapt this approach to our pairs of tiles (\Cref{lem:type2-pre}). 

It is easy to arrange the tiles in $T_{\theta,\lambda}$, $\lambda\in \Lambda$, into groups in which the scales differ by integer factors---so we can obtain a laminar system, which supports a quadtree. However, for tiles in a standard quadtree, the translated tiles in $T'_{\theta,\lambda}$ would not be laminar. For this reason, we use a variation of quadtrees: nonatrees~\cite{BhoreNTW26}, in which each tile is subdivided into nine congruent subtiles. Using nonatrees, the tiles in $T_{\theta,\lambda}$ and $T'_{\theta,\lambda}$, over all $\lambda\in \Lambda$, can be grouped into laminar systems.  

\paragraph{Edges of Type~1.}
We introduce a charging scheme: We charge each edge $st$ of Type~1 to one of its vertices as follows. Assume $st$ was added to $G$ when processing the pair tiles $T_{\theta,\lambda}(i,j)$ and $T_{\theta,\lambda}(i+\mu+1,j)$. If $N_{\theta,\lambda}(i,j)$ is monochromatic, we charge $st$ to $s$, otherwise (when $N_{\theta,\lambda}(i+\mu+1,j)$
is monochromatic), we charge $st$ to $t$. 

\begin{lemma}\label{lem:type1-pre}
Let $s\in S$. Assume that $s$ lies in tiles $T_{\theta,\lambda}(i,j)$ and $T_{\theta,\lambda'}(i',j')$ with monochromatic neighborhoods, where $\lambda = 3^{p/(\mu+2)}$ and $\lambda' = 3^{p'/(\mu+2)}$ with $p<p'$, and our construction added an edge $st$ of Type~1 when processing the pair of tiles $T_{\theta,\lambda}(i,j)$ and $T_{\theta,\lambda}(i+\mu+1,j)$.
Then we have $p'< p+ (\mu+2)\log_3 (\mu+3)$.
\end{lemma}
\begin{proof}
   Assume w.l.o.g.\ that $\theta=0$ and $s$ is red (i.e., $s\in R)$. Since we added an edge $st$ for some $t\in N_{\theta,\lambda}(i+\mu+1,j)$, then $t$ is blue. 
   Since $s\in N_{\theta,\lambda}(i,j)$ and $t\in N_{\theta,\lambda}(i+\mu+1,j)$,
   where width of $T_{\theta,\lambda}(i,j)$ is $\lambda$, then 
   $x(s)\geq (i-1)\lambda$ and $x(t)\leq (i+\mu+2)\lambda$. Consequently, 
\begin{align*}
  |x(s)-x(t)|&\leq |(i-1)\lambda - (i+\mu+2)\lambda| 
   =(\mu+3)\lambda\\
|y(s)-y(t)|&\leq \delta\lambda.
\end{align*}
Recall that $s\in T_{\theta,\lambda'}(i,'j')$. For any $q\in \mathbb{R}^2$ such that $|x(s)-x(q)|\leq \lambda'$ and $|y(s)-y(q)|\leq \delta\cdot \lambda'$, we have $q\in N_{\theta,\lambda'}(i,'j')$. Hence $t\in N_{\theta,\lambda'}(i,'j')$ if 
\begin{align*}
  (\mu+3)\lambda &\leq \lambda'\\
    (\mu+3)3^{p'/(\mu+2)} &\leq 3^{p'/(\mu+2)}\\
    (\mu+3) &\leq 3^{(p'-p)/(\mu+2)}\\
p+(\mu+2)\log_3 (\mu+3) &\leq p'. 
\end{align*}
Hence if $p'\geq p+ (\mu+2)\log_3 (\mu+3)$, then the neighborhood $N_{\theta,\lambda'}(i,'j')$ is bichromatic, and $st$ cannot be charged to $s$. 
We conclude that $p'< p+ (\mu+2)\log_3 (\mu+3)$, as claimed. 
\end{proof}

\begin{lemma}\label{lem:type1}
For every bichromatic set $S=R\cup B$ of $n$ points in the plane, and for every $\theta\in \Phi$, the graph $G$ has $O(n)$ edges of Type~1.
\end{lemma}
\begin{proof}
    Each vertex $s\in S$ receives charges from at most 26 scales by \Cref{lem:type1-pre}. Consequently, there are $O(|S|)=O(n)$ edges of Type~1.
\end{proof}

\paragraph{Edges of Type~2.} 
For the purpose of analyzing the edges of Type~2, we arrange the nonempty tiles $T_{\theta,\lambda}(i,j)$ into $\mu+2$ nonatrees as follows. For $\sigma=0,1,\ldots , \mu+1$, let 
    \[
    \Lambda_\sigma 
    = \left\{3^{\sigma/(\mu+2)}\cdot 3^p: p\in \mathbb{Z}\right\}.
    \]
Then the set of tiles $\mathcal{T}_\sigma=\{T_{\theta,\lambda}(i,j):  i,j\in \mathbb{Z}, \lambda\in \Lambda_\sigma\}$ is a laminar set system (i.e., any two tiles are either interior-disjoint or one contains the other). Importantly, the translated tiles $\mathcal{T}'_\sigma=\{T'_{\theta,\lambda}(i,j): i,j\in \mathbb{Z}, \lambda\in \Lambda_\sigma\}$ also form a laminar system~\cite{BhoreNTW26} (this is why we use nonatrees, instead of quadtrees). 
Let $\mathcal{Q}_\sigma$ denote the set of tiles $T_{\theta,\lambda}(i,j)\in \mathcal{T}_\sigma$ such that $T_{\theta,\lambda}(i,j)\cap S\neq \emptyset$ (the case of the tiles $T'_{\theta,\lambda}(i,j)\in \mathcal{T}'_\sigma$ is analogous). Containment among the tiles in $\mathcal{Q}_\sigma$ defines a hierarchy (i.e., poset or tree); and we refer to $\mathcal{Q}_\sigma$ as a \emph{nonatree}. We can now define a parent-child relationship between the tiles in $\mathcal{T}_\sigma$:  For $T_{\theta,\lambda}(i,j),T_{\theta,\lambda'}(i',j')\in \mathcal{T}_\sigma$, we say that $T_{\theta,\lambda}(i,j)$ is a \emph{child} of $T_{\theta,\lambda'}(i',j')$ if $T_{\theta,\lambda}(i,j)\subset T_{\theta,\lambda'}(i',j')$ and $\lambda'=3\lambda$.
It is clear from the definition that every tile has at most nine children. 
Let $\mathcal{R}_\sigma$ be the set of tiles $T\in \mathcal{Q}_\sigma$ that have two or more children. Containment among the tiles in $\mathcal{R}_\sigma$ gives a partition tree on $S$ for every $\sigma\in \{0,\ldots , \mu+1\}$, consequently 
\begin{equation}\label{eq:partition}
|\mathcal{R}_\sigma|\leq n.
\end{equation}
Note that every tile $T\in \mathcal{Q}_\sigma$ is contained in some tile $\widehat{T}\in \mathcal{R}_\sigma$; possibly $\widehat{T}=T$. For every $T\in \mathcal{Q}_\sigma$, let $\overline{p}(T)$ be the smallest tile in $\mathcal{R}_\sigma$ that contains $T$ (i.e., $\overline{p}(T)$ is the closest ancestor of $T$ in $\mathcal{R}_\sigma$). 

We can now introduce a charging scheme: Consider the edges $r^*b^*$ and $r^{**}b^{**}$ of Type~2 that were added to $G$ when we processed the pair of tiles $\{T_1,T_2\}=\{T_{\theta,\lambda}(i,j), T_{\theta,\lambda}(i+\mu+1,j)\}$ with bichromatic neighborhoods. Let $\lambda_1$ and $\lambda_2$ be the scales of $\overline{p}(T_1)$ and $\overline{p}(T_2)$, respectively. If $\lambda_1\leq \lambda_2$, then we \emph{charge} both edges to $\overline{p}(T_1)$; otherwise (if $\lambda_1>\lambda_2)$, we charge them to $\overline{p}(T_2)$. 

\begin{lemma}\label{lem:type2-pre}
In the charging scheme above, with constant separation parameter $\mu\geq 1$, every tile $T\in \bigcup_{\sigma=0}^{\mu+1}\mathcal{R}_\sigma$ receives $O(1)$ charges. 
\end{lemma}
\begin{proof}
Let $\sigma\in \{0,\ldots ,\mu+1\}$ and $T\in \mathcal{R}_\sigma$. 
Assume that $T$ has received $c$ units of charges.  Assume w.l.o.g.\ that $T=T_{\theta,\lambda}(i,j)$ with $\theta=\sigma=0$ and $\lambda=1$. In particular, $T$ is an axis-aligned rectangle of $\width(T)=1$.  

Let $\mathcal{P}_1$ be the collection of pairs of tiles $(T_1,T_2)$ such that two bichromatic edges corresponding to $(T_1,T_2)$ are charged to $T$ and $\overline{p}(T_1)=T$. By assumption, we have $|\mathcal{P}_1|\geq c/2$. 
In each pair $(T_1,T_2)\in \mathcal{P}_1$, the tile $T_1$ is either \emph{left} (i.e., $T_1=T(i,j)$ and $T_2=T(i+\mu+1,j)$) or \emph{right} (i.e., $T_1=T(i+\mu+1,j)$ and $T_2=T(i,j)$). We may assume w.l.o.g.\ that in at least half of these pairs, $T_1$ is left. Let $\mathcal{P}_2\subseteq \mathcal{P}_1$ be the set of these pairs; note that $|\mathcal{P}_2|\geq  |\mathcal{P}_1|/2 \geq c/4$. 

For each pair $(T_1,T_2)\in \mathcal{P}_2$, either $T_1=T$ or the nonatree $\mathcal{Q}_\sigma$ contains a descending chain $\mathrm{chain}(T_1)$ from $T$ to $T_1$ where all intermediate tiles have only one child (i.e., they are in $\mathcal{Q}_\sigma\setminus \mathcal{R}_\sigma$). There is at most one pair $(T_1,T_2)\in \mathcal{P}_2$ such that $T_1=T$; so $T_1$ is a strict descendant of $T$ in at least $c/4-1$ pairs in $\mathcal{P}_2$. 
If $T_1\neq T$, then $\mathrm{chain}(T_1)$ includes exactly one child of $T$. Since $T$ has at most nine children in $\mathcal{Q}_\sigma$, there are at least $(c/4-1)/9>c/36-1$ pairs in $\mathcal{P}_2$ where $\mathrm{chain}(T_1)$ includes the same child of $T$; let $T_{\rm child}$ denote this child of $T$ and let $\mathcal{P}_3$ be the set of these pairs; see \Cref{fig:quadtree}.

 \begin{figure}[tbh]
        \centering
        \includegraphics[width=.4\textwidth]{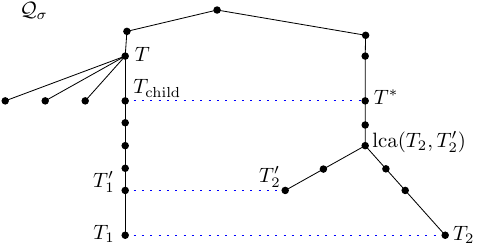}
        \caption{A schematic layout of the nodes in the nonatree $\mathcal{Q}_\sigma$}
        \label{fig:quadtree}
    \end{figure}

Note for all pairs $(T_1,T_2)\in \mathcal{P}_3$, the tiles $T_1$ are at different scales (since they are in the same single-child chain). Let $(T_1,T_2)$ and $(T_1',T_2')$ be the pairs in $\mathcal{P}_3$ at the lowest and the second lowest scales, respectively. Since $|\mathcal{P}_3|> c/36-1$ and $\width(T)=1$, 
then $\width(T_1)< 3^{1-c/36}$ and $\width(T_1')< 3^{1-c/36}$. Since $T_1\subseteq T_1'$, then 
\begin{equation}\label{eq:type2}
\dist(T_1',T_2)\leq \dist(T_1,T_2)=\mu\,\width(T_1)\leq \frac{\mu}{3}\, \width(T_1'). 
\end{equation}
It follows that 
\begin{align*}
\dist(T_2,T_2')
    &\geq \dist(T_1',T_2')-\dist(T_1',T_2) -\width(T_2)\\
    &\geq \dist(T_1',T_2')-\dist(T_1,T_2) -\width(T_2)\\
    &\geq \left(\mu-\frac{\mu}{3}-\frac13\right)\, \width(T_1') \\
    &=\frac{2\mu-1}{3}\,\width(T_1')>0,
\end{align*} 
and, in particular, the tiles $T_2$ and $T_2'$ are disjoint for all $\mu\geq 1$. 
Furthermore, $T_1\subset T_1'\subseteq T_{\rm child}$ implies 
\begin{align*}
\dist(T_{\rm child},T_2)&\leq \dist(T_1,T_2) = \mu\, \width(T_1) \leq \mu\cdot 3^{1-c/36} ,\\
\dist(T_{\rm child},T_2')&\leq \dist(T_1',T_2')=\mu\, \width(T_1') \leq \mu\cdot 3^{2-c/36} .
\end{align*}
As $\width(T_{\rm child})=\frac13$, both $T_2$ and $T_2'$ are contained in some tile $T^*\in \mathcal{Q}_\sigma$ at the same scale as $T_{\rm child}$
if 
\begin{align*}
    \mu\cdot 3^{2-c/36} &\leq \frac13\\
\log_3\mu + (2-c/36) &\leq -1\\
        36 \left(3+\log_3 \mu\right)&\leq c.
\end{align*}
 But $\mathrm{chain}(T_1)$ is a single-child chain of tiles, so  $T_2,T_2'\not\subseteq T_{\rm child}$, and $T^*\neq T_{\rm child}$ (although, $T^*$ could be a sibling of $T_{\rm child}$). 
However, this would lead to a contradiction: the least common ancestor $\mathrm{lca}(T_2,T_2')$ of $T_2$ and $T_2'$ in the nonatree $\mathcal{Q}_\sigma$ satisfies $\mathrm{lca}(T_2,T_2')\subseteq T^*$. In particular, $\mathrm{lca}(T_2,T_2')$ is at a scale below $T$; and $\mathrm{lca}(T_2,T_2')$ has at least two children, which implies $\mathrm{lca}(T_2,T_2')\in \mathcal{R}_\sigma$. Consequently, both $\overline{p}(T_2)$ and $\overline{p}(T_2')$ are at a scale below $T=\overline{p}(T_1)=\overline{p}(T_1')$; and so our charging scheme could not have charged $(T_1,T_2)$ and $(T_1',T_2')$ to $T$. We conclude that every tile $T$ receives at most $c<36 \left(3+\log_3 \mu\right)$ units of charges.
\end{proof}

\begin{lemma}\label{lem:type2}
For every bichromatic set $S=R\cup B$ of $n$ points in the plane, and for every $\theta\in \Phi$, the graph $G$ has $O(n)$ edges of Type~2.
\end{lemma}
\begin{proof}
This readily follows from the combination of \Cref{eq:partition} and \Cref{lem:type2-pre}.
\end{proof}

\subsection{Running Time Analysis}
\label{ssec:runtime}

For every set $S$ of $n$ bichromatic points in the plane and every $\eps>0$, our spanner construction (\Cref{subsec:construct}) can be implemented in $O(\sqrt{1/\eps}\cdot n\log n)$ time. The construction is done independently for each $\theta\in \Phi$, where $|\Phi|=O(\sqrt{1/\eps})$, so it is enough to show that the construction takes $O(n\log n)$ time for each $\theta\in \Phi$. Assume w.l.o.g.\ that $\theta=0$. Instead of infinite tilings, we can use nonatrees (the generalization of quadtrees, where each cell is split into 9 congruent cells) after an affine transformation $f(x,y)=(x,y/\delta)$, which  maps a $1\times \delta$ rectangle into a unit square. Let $S_0:=f(S)$. Let $Q_0$ be an axis-aligned bounding box of $S_0$ where $\width(Q_0)=2\, \mathrm{height}(Q_0)$, and scale $Q_0$ centrally up by factors of $3^{\sigma/11}$ for $\sigma=0,\ldots ,10$. Each scaled copy of $Q_0$ can be expressed as an intersection of two axis-aligned squares $Q_0=Q\cap Q'$, where $\width(Q)=\width(Q')=\width(Q_0)$; note that $Q'=Q+(0,\frac12\, \mathrm{height}(Q))$.

Starting from each of these 22 squares ($Q$ and $Q'$ with $\sigma=0,\ldots  ,10$), as root cells, we can build compressed nonatrees $\mathcal{R}_\sigma$ and $\mathcal{R}'_\sigma$ for $S_0$ in $O(n\log n)$ time; analogously to the construction of quadtrees~\cite{BhoreNTW26,BergCKO08,HarPeled11}. A compressed nonatree for $n$ points has $O(n)$ cells, but the (uncompressed) nonatree can have arbitrarily many cells (as it may have descending chains of arbitrarily length). By \Cref{lem:type1-pre,lem:type2-pre}, our construction only uses (uncompressed) nonatree nodes that have an ascending chain of length at most 5 to a node of the compressed nonatree $\mathcal{R}_\sigma$: We augment $\mathcal{R}_\sigma$ with up to 5 single-child descendants per node to obtain a tree $\mathcal{R}^*_\sigma$ with $O(n)$ nodes. In a bottom-up traversal of $\mathcal{R}^*_\sigma$, we can find the leftmost and rightmost red (blue) points in each cell in $O(n)$ time. For each cell $T(i,j)$, we can find their eight adjacent cells in $N(i,j)$ in $O(\log n)$ time, and classify $N(i,j)$ as either monochromatic or bichromatic. In another bottom-up traversal, we find all pairs of tiles $T(i,j)$ and $T(i+10,j)$ and find the edges of $G$, according to Cases~1--3 in \Cref{subsec:construct}, in $O(n)$ time. 

\subsection{Generalization to Multichromatic Point Sets}
\label{ssec:multichromatic}

Our construction for the $(3+\eps)$-spanner can be generalized for multichromatic point sets: Let $S$ be a set of $n$ multichromatic points in the plane, which is partitioned into $k\geq 2$ color classes $S=V_1\cup \ldots\cup V_k$. For a given parameter $\eps > 0$, we use the family of tilings $\mathcal{T}$ defined in \Cref{subsec:tiling} with the separation constant $\mu = 12$. 

Recall that for bichromatic point sets, the construction in \Cref{subsec:construct} distinguishes between four cases based on whether the neighborhoods $N_{\theta, \lambda}(i,j)$ and $N_{\theta, \lambda}(i+10,j)$ are monochromatic or bichromatic. For multichromatic point sets, we follow a modified version of the construction in \Cref{subsec:construct} and distinguish between four cases based on whether the neighborhoods $N_{\theta, \lambda}(i,j)$ and $N_{\theta, \lambda}(i+13,j)$ are monochromatic or multichromatic.

We begin by initializing $G =(S,E)$ to be the empty graph and consider four cases. In Cases~0 and 1, both $N_{\theta, \lambda}(i,j)$ and $N_{\theta, \lambda}(i+13,j)$ are monochromatic, and we handle them exactly as in \Cref{subsec:construct} (we add a bistar between them if they have different colors). Cases~2 and 3 are replaced by the following modified versions:

\begin{enumerate}
    \item[Case~2:] Exactly one of $N_{\theta, \lambda}(i,j)$ and $N_{\theta, \lambda}(i+13,j)$ is monochromatic. Assume w.l.o.g. that $N_{\theta, \lambda}(i,j)$ is monochromatic and $N_{\theta, \lambda}(i+13,j)$ is multichromatic (the other case is symmetric). Let $v^*$ be a leftmost point in $S\cap N_{\theta, \lambda}(i+13,j)$ that is of different color than all points in $N_{\theta, \lambda}(i,j)$ and for every point in $S\cap T_{\theta, \lambda}(i,j)$ add an edge to $v^*$. 

    \item[Case~3:] Both $N_{\theta, \lambda}(i,j)$ and $N_{\theta, \lambda}(i+13,j)$ are multichromatic. Let $r^*$ be a rightmost point in $S\cap N_{\theta, \lambda}(i,j)$, and let $b^{**}$ be a rightmost point in $S\cap N_{\theta, \lambda}(i,j)$ of a color  other than the color of $r^*$. We may assume w.l.o.g.\ that $r^*$ is red and $b^{**}$ is blue. Let $u^*$ and $v^{**}$ be leftmost points in $S\cap N_{\theta, \lambda}(i+13,j)$ that is non-red and non-blue, respectively. In this case, we add two edges to $G$: the edges $r^*u^*$ and $b^{**}v^{**}$.
\end{enumerate}

This completes the construction of the $k$-partite graph $G = (S,E)$. The following lemma shows that $G$ is a $(3+\eps)$-spanner of the complete multipartite graph $K(V_1,\ldots,V_k)$. Unlike the bichromatic argument, the stretch analysis for Cases~2--3 will be more complicated, and will each require distinguishing between two subcases. 

\begin{lemma} \label{lem:genstretch}
    For every finite multichromatic set $S$ in the plane, the graph $G \subseteq K(V_1,\ldots,V_k)$ constructed above is a $(3+\eps)$-spanner for $K(V_1,\ldots,V_k)$.
\end{lemma}

\begin{proof}
    We proceed by induction on the length of multichromatic edges in $K(V_1,\ldots,V_k)$. In particular, let $
    \alpha,\beta \in \{1,2,\ldots,k\}$ and $D = \bigcup_{\alpha\ne \beta}\{|uv| : (u,v) \in V_\alpha \times V_\beta\}$ be the set of distances of multichromatic pairs in $S$; and let $0 = d_0 < d_1 < d_2 < \cdots < d_m$ be elements of $\{0\} \cup D$ in increasing order. 

    We prove, by induction on $\ell = 0,1,\ldots, m$, for all $\alpha \ne \beta$, $(u,v) \in V_\alpha \times V_\beta$ with $|uv| \le d_\ell$, the graph $G$ contains an $uv$-path of length at most $(3+\eps)|uv|$. In the base case, where $\ell = 0$ and $d_\ell = 0$, the claim vacuously holds, as there is no multichromatic pair $(u,v)$ with $|uv| = 0$.

    In the induction step, let $d_\ell \in D$ and assume that for all $(u',v') \in V_\alpha \times V_\beta$, if $|u'v'| < d_\ell$, then $G$ contains an $u'v'$-path of length at most $(3+\eps)|u'v'|$. Let $(u,v) \in V_\alpha \times V_\beta$ with $|uv| = d_\ell$. By \Cref{lemma:spacing}, there exist $\theta \in \Phi$ and $\lambda \in \Lambda$, $i,j \in \mathbb{Z}$ such that $T_{\theta,\lambda}(i,j)$ and $T_{\theta, \lambda}(i+13,j)$ each contains one of $u$ and $v$; or $T'_{\theta,\lambda}(i,j)$ and $T'_{\theta, \lambda}(i+13,j)$ each contains one of $u$ and $v$. To simplify our argument, we may assume w.l.o.g. that $\theta = 0$, $\lambda = 1$, $u \in T_{\theta, \lambda}(i,j)$ and $v \in T_{\theta,\lambda}(i+13, j)$. Under these assumptions, $N_{\theta,\lambda} = N(i,j)$, which is an axis-aligned rectangle of width 3 and height $3\delta$. We distinguish between three cases:

    \begin{enumerate}
        \item[Case~1:] The upper bound on the stretch factor follows similarly from \Cref{eq:case1} in \Cref{lem:stretch} with the new separation constant $\mu = 12$. We have that the length of the $uv$-path is at most $(3+\eps)|uv|$.
        \item[Case~2:] One of $N(i,j)$ and $N(i+13,j)$ is monochromatic and the other is multichromatic. Assume w.l.o.g.\ that $N(i,j)$ is monochromatic and $N(i+13,j)$ multichromatic (the other case is symmetric). We distinguish between two subcases:
        
        Subcase~2a: $v^*$ and $v$ are of different colors. Since $v^*,v\in N(i+13,j)$, the edge $v^*v$ lies in an axis-aligned rectangle spanned by $T(i+13,j)$ and an adjacent tile, which implies $|v^*v| \le 2\diam(T(i,j))$. Recall that $\diam(T(i,j)) = \sqrt{1+\delta^2} \le 1 + \eps/14$ by the Pythagorean Theorem and \Cref{eq:P}. We have that $|v^*v| \le 2\diam(T(i,j)) \le 2(1 + \eps / 14) < 12 \le |uv|$. By the induction hypothesis, $G$ contains some path $P = (v^*,\ldots,v)$ of length at most $(3+\eps)|v^*v|$. Then, the concatenation $uv^* \oplus P$ is a $uv$-path in $G$ and its length can be bounded above as follows
        \begin{align*}
            |uv^*| + |P| &\le \sqrt{14^2 + (2\delta)^2} + (3+\eps)|v^*v| \\
            &\le 14 + 2\delta^2 + 2(3+\eps)\diam(T(i,j)) \\
            &\le 14 + 2\eps / 7 + 2(3+\eps)(1+\eps / 14) \\
            &\le 20 + 20\eps / 7 \\
            &< 12(3+\eps) \\
            &\le (3+\eps)|uv|.
        \end{align*}
        So the concatenation $uv^* \oplus P$ is a $uv$-path of length at most $(3+\eps)|uv|$.
        
        Subcase~2b: $v^*$ and $v$ are of the same color. This case is identical to Case~2 in \Cref{lem:stretch} and the upper bound on the stretch factor follows similarly from \Cref{eq:case2} with the new separation constant $\mu = 12$. Thus, $G$ contains a $uv$-path of length at most $(3+\eps)|uv|$.

        \item[Case~3:]  Both $N(i,j)$ and $N(i+13,j)$ are multichromatic. Assume w.l.o.g.\ that $u$ is non-blue (the other case is symmetric). We distinguish between two subcases: 
        
        Subcase~3a: $v^{**}$ and $v$ are of different colors. This case is identical to Case~3 in \Cref{lem:stretch} and the upper bound on the stretch factor follows similarly from \Cref{eq:case3} with the new separation constant $\mu = 12$. Thus, $G$ contains a $uv$-path of length at most $(3+\eps)|uv|$.

        \begin{figure}[tbh]
        \centering
        \includegraphics[width=.95\textwidth]{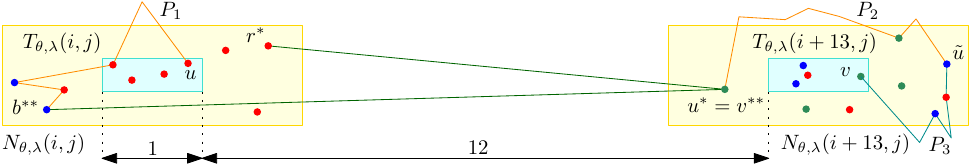}
        \caption{A $uv$-path when both $N(i,j)$ and $N(i+13, j)$ are multichromatic}
        \label{fig:stretch3}
    \end{figure}

        Subcase~3b: $v^{**}$ and $v$ are of the same color; see \Cref{fig:stretch3}. Since $N(i+13,j)$ is multichromatic, there exists a point $\tilde{u} \in S\cap N(i+13,j)$ of a color different from that of $v$ and $v^{**}$. Similar to Case~2, we have that $|ub^{**}| \le 2\diam(T(i,j)) < |uv|$ and $|\tilde{u}v| \le 2\diam(T(i+13,j)) < |uv|$. Since $v^{**},\tilde{u}\in N(i+13,j)$, the edge $v^{**}\tilde{u}$ must have length at most $3\diam(T(i+13, j)) = \sqrt{3^2 + (3\delta)^2} < 12 \le |uv|$. 
        By the induction hypothesis, $G$ contains some paths 
        \[P_1 = (u,\ldots,b^{**}) \quad\text{and}\quad P_2 = (v^{**},\ldots,\tilde{u}) \quad\text{and}\quad P_3 = (\tilde{u},\ldots,v),\]
        resp., of length at most $(3+\eps)|ub^{**}|$, $(3+\eps)|v^{**}\tilde{u}|$, and $(3+\eps)|\tilde{u}v|$. Then, the concatenation $P_1 \oplus b^{**}v^{**} \oplus P_2 \oplus P_3$ is a $uv$-path of $G$. Since $v^{**}$ and $v$ are of the same color, where $v\in T(i+13,j)$ and $v^{**}$ is a leftmost point in $S\cap N(i+13,j)$, then $|x(b^{**})-x(v^{**})|\leq 15$. Consequently, the length of this $uv$-path can be bounded above as follows
        \begin{align*}
            |P_1| + |b^{**}v^{**}| + |P_2| + |P_3| &\le (3+\eps)|ub^{**}| + \sqrt{15^2 + (3\delta)^2} + (3+\eps)|v^{**}\tilde{u}| + (3+\eps)|\tilde{u}v| \\ 
            &\le 7(3+\eps)\diam(T(i,j)) + 15 + 9\eps/14 \\
            &\le 36 + 135\eps/14 \\
            &< 12(3+\eps) \\
            &\le (3+\eps)|uv|.
        \end{align*}
        So the concatenation $P_1 \oplus b^{**}v^{**} \oplus P_2 \oplus P_3$ is a $uv$-path of length at most $(3+\eps)|uv|$. This completes the proof in both subcases. 
    \end{enumerate}

    Overall, the graph $G$ is a $(3+\eps)$-spanner of the complete $k$-partite graph $K(V_1,\ldots,V_k)$.
\end{proof}

\paragraph{Sparsity Analysis.} The multichromatic construction added the same number of edges in each case as the bichromatic analysis. The same charging schemes for edges of Type~1 (\Cref{lem:type1-pre}) and Type~2 (\Cref{lem:type2-pre}) with separation parameter $\mu=12$ yield an upper bound $O(n)$ on the number of edges. 

\subsection{Generalization to Higher Dimensions}
\label{ssec:dspace}

It is straightforward to generalize our construction to Euclidean $d$-space: 
In $\mathbb{R}^2$, we used basic tiles of size $1\times \sqrt{\eps/7}$, and rotated the tilings with a set $\Phi\subset [0,\pi]$ of $\Theta(\sqrt{1/\eps})$ angles at distance $\arctan(\sqrt{\eps/7}/22)=O(\sqrt{\eps})$ apart. In $d$-space, for $d\geq 3$, we use basic tiles of size $1\times \sqrt{\eps/(d+5)} \times \ldots \times \sqrt{\eps/(d+5)}$, and we use a set $\Phi$ of rotations constructed as follows: We choose a set $U_d\subset \mathbb{S}^{d-1}$ of unit vectors such that the $O(\sqrt{\eps})$-neighborhoods of the vectors in  $U_d$ (where each neighborhood is a spherical cap) cover $\mathbb{S}^{d-1}$. A standard packing argument shows that $O(\eps^{(1-d)/2})$ directions suffice. Then for each $\vec{u}\in U_d$, we rotate the basic tiling such that the positive $x_1$-axis is parallel to $\vec{u}$. With this setup, \Cref{lemma:spacing} generalizes. The construction, as well as the stretch and sparsity analyses readily extend to $\mathbb{R}^d$. 
The number of edges is $O(|U_d|\cdot  n)=O(\eps^{(1-d)/2}n)$.

\paragraph{Doubling Spaces.} We note that our construction also generalizes to metric spaces with bounded doubling dimension. Since the scope of this paper pertains to Euclidean spaces, we give only a very brief outline. Given a bichromatic point set $S=R\cup B$ of $n$ points in a metric space with constant doubling dimension $\mathrm{ddim}$, we can construct a WSPD of size $O(n)$ for $S$, with separation constant about 10~\cite{KrauthgamerL04,Talwar04} (see also \cite{HarPeledM06}). 
We can construct a graph $G=(S,E)$ by adding edges for each well-separated pair as in \Cref{subsec:construct}. The stretch and sparsity analyses (\Cref{subsec:analysis,subsec:analysis2}) go through. However, we cannot take advantage of the ``shape'' of the tiles, and the dependence of $\eps$ will be weaker: We obtain a $(3+\eps)$-spanner with  $O(n/\eps^{\mathrm{ddim}-1})$ edges.

\section{Lower Bound for Bichromatic Spanners in the Plane}
\label{sec:lower}

In this section, we establish a lower bound for the size of a $(3+\eps)$-spanner of $n$ bichromatic points in the plane. We analyze the following bichromatic point set: Let $R$ and $B$, resp., be set of points on the left and right sides of the axis-aligned unit square $[0,1]^2$.  Let $S = R \cup B$, and let $K(R,B)$ be the complete bipartite graph on $S$. We observe that every edge $e$ of $K(R,B)$ satisfies 
\begin{equation}\label{eq:edgebd}
    1 \le |e| \le \sqrt{2} .
\end{equation}
Next, we introduce a few preliminary lemmas.

\begin{lemma} \label{lem:onethree}
    Let $G=K(R,B)$ and $(r,b) \in R \times B$. Every $rb$-path of length at most $(3+\eps)|rb|$  consists of one or three edges if $\eps < \frac{5}{\sqrt2}-3\approx 0.535$.
\end{lemma}

\begin{proof}
    Since the graph $G$ is a bipartite graph, any $rb$-path must contain an odd number of edges. Using \Cref{eq:edgebd} and $\eps < \frac{5}{\sqrt2} - 3$, we have 
    \[(3+\eps)|rb| \le (3+\eps)\sqrt{2} < \left(3+\frac{5}{\sqrt2} - 3\right)\sqrt{2} = 5 \le 5 \cdot \min |e|.\]
    Consequently, any $rb$-path of length at most $(3+\eps)|rb|$ must consist of one or three edges. 
\end{proof}

\begin{lemma}\label{lem:y9h}
    Let $h\in [\sqrt{\eps},\frac13]$ and 
    $(r,b) \in R \times B$ such that $|y(r) - y(b)| < h$, 
    where $y(r)$ and $y(b)$, resp., denote the $y$-coordinate of $r$ and $b$. Then every vertex $v$ of any $rb$-path of length at most $(3+\eps)|rb|$ satisfies 
    \[
    \min\big\{|y(v) - y(r)|,|y(v) - y(b)| \big\} <3h .
    \]
\end{lemma}

\begin{proof}
Let $P$ be an $rb$-path of length at most $(3+\eps)|rb|$. By \Cref{lem:onethree}, $P$ has one or three edges. The conclusion holds vacuously if $P$ has one edge only. We may assume that $P=(r,b',r',b)$. W.l.o.g. assume that $v = b'$.
We claim that $|y(b') - y(r)| < 3h$.
Suppose for contradiction $|y(r) - y(b')| \ge 3h$. Then, by the Pythagorean Theorem, the length of $rb'$ can be bounded below by $|rb'| \ge \sqrt{1+ (3h)^2}$. Similarly, the length of $rb$ can be bounded above by $|rb| \le \sqrt{1+h^2}$. Using \Cref{eq:edgebd}, we have $|b'r| \ge 1$ and $|r'b| \ge 1$. It follows that 
    \[\sqrt{1+(3h)^2} + 2 \le |P| \le (3+\eps)|rb| \le (3+\eps)\sqrt{1+h^2} .\]
Using the Taylor estimates $1+\frac{x}{2}-\frac{x^2}{8}\leq \sqrt{1+x}\leq 1+\frac{x}{2}$, as well as $h\leq \frac13$ and $\eps< 1$, we obtain
\begin{align*}
 1+\frac{(3h)^2}{2}-\frac{(3h)^4}{8} + 2 &\le (3+\eps)(1+h^2/2)\\ 
 3+\frac{(3h)^2}{2}-\frac{(3h)^4}{8}  & \le 3+\eps + (3+\eps)\frac{h^2}{2}\\ 
\frac{(3h)^2}{2}-\frac{(3h)^2}{8} &\le \eps +2h^2\\ 
  \frac{11h^2}{8} &\le \eps \\ 
 h&\le \sqrt{8\eps/11} ,
\end{align*}
which contradicts the assumption that $h\geq \sqrt{\eps}$.
\end{proof}

\subsection{Lower Bound Construction}
\label{ssec:LB}

\paragraph{Basic Lower Bound Construction.}
Let $\eps \in (0,\frac19)$ be given. We first present a basic construction for $n_0 = 2\, \left\lfloor 1/(3\sqrt\eps)\right\rfloor$, and later generalize it to arbitrary $n=\Omega(1/\sqrt\eps)$ by using $\Theta(\sqrt\eps n)$ disjoint copies of the basic construction.  Let the set $R$ (resp., $B$) consist of $n_0/2$ red (resp., $n_0/2$ blue) points evenly distributed on the left (right) side of the unit square $[0,1]^2$; and let $S=R\cup B$. Let $G=(S,E)$ be a bichromatic $(3+\eps)$-spanner for $S=R\cup B$. 

\begin{observation}\label{obs:height}
    For every interval $I\subseteq [0,1]$ of length $|I|=h$, the axis-aligned box $[0,1]\times I$ contains $\lfloor h/(3\sqrt{\eps})\rfloor$ red points (and the same number of blue points).
\end{observation}

\paragraph{Sparsity Analysis.}
We show that $G=(S,E)$ has $\Omega(\frac{1}{\sqrt\eps}\log \log \frac{1}{\eps})$ edges. We recursively partition $E$ into $\Omega(\log \log \frac{1}{\eps})$ subsets (\emph{levels}), and show that  each subset contains $\Omega(1/\sqrt{\eps})=\Omega(n_0)$ edges. 

Specifically, we recursively construct an increasing sequence of real numbers $(h_0=0, h_1,h_2,\ldots ,h_m)$, and for every $i=0,\ldots, m$, we define  
\[
    E_i=\{ e\in E: |\mathrm{slope}(e)|\in [h_i,h_{i+1})\},
\]
where the slope of a nonvertical line segment $e=ab$ is $\mathrm{slope}(e)=\frac{y(b)-y(a)}{x(b)-x(a)}$. Note that the sets $E_0,\ldots , E_m$ are disjoint, hence 
    \begin{equation}\label{eq:partition2}
    |E|\geq \sum_{i=0}^m|E_i|.
    \end{equation}

Initially, let $h_0=0$ and $h_1=3\sqrt{\eps}$. Applying \Cref{lem:y9h} with $h=\sqrt{\eps}$, for all $(r,b)\in R\times B$ with $y(r)=y(b)$, $G$ does not contain any 3-edge $rb$-path of length at most $(3+\eps)$, and so $E$ must contain the edge $rb$. There are $n_0/2$ such edges, and they all have slope  0, consequently $|E_0|\geq n_0/2$.  

In step $i$ of the recursion, assume that $h_{i}<\frac13$ is already defined. We need to define $h_{i+1}$, and show that $|E_i|\geq \Omega(n_0)$.  The recursion terminates when $h_{i+1}>\frac13$ (because \Cref{lem:y9h} is not applicable for $h>\frac13$). 

Let $h_{i+1} = 7(2t_i +4)h_{i}$ with $t_i = 2h_{i} / (3\sqrt\eps)$. We construct a packing of $C$ of the unit square $[0,1]^2$ with interior-disjoint  axis-aligned rectangles of size $1\times h_{i+1}$: 
\[
C = \bigcup_{0 \le j < \lfloor 1/h_{i+1}\rfloor} [0,1] \times 
\big[jh_{i+1}, (j+1)h_{i+1}\big).
\]

\begin{lemma}\label{lem:manym}
    We have $m\geq\Omega(\log \log \frac{1}{\eps})$.
\end{lemma}
\begin{proof}
Observe that $\{t_i\}_{i=1}^m$ is an increasing sequence with $t_1 = 2$, we can bound $t_{i+1}$ from above by a quadratic monomial in $t_i$ as follows
\[
    t_{i+1} = \frac{2h_{i+1}}{3\sqrt{\eps}} 
    =\frac{14(2t_i + 4)h_i}{3\sqrt\eps} 
    =14t_i^2 + 28t_i 
    \le 28t_i^2.
\]
By induction on $i$, this gives $t_i< 28^{2^{i-1}}$ for $i=1,\ldots , m$. 
Recall that there are at most $1/(3\sqrt\eps)$ red (resp., blue) points evenly distributed on the left (right) side of the unit square. 
From the definition of $\{h_i\}_{i=0}^m$, we have $h_m\leq \frac13<h_{m+1}$,
hence $t_{m+1} = 2h_{m+1} / (3\sqrt\eps) > 2/ (9\sqrt\eps) $.
Contrasting the upper and lower bounds for $t_{m+1}$, we obtain  
\begin{align*}
    \frac{2}{9\sqrt\eps} & < 28^{2^{m}}\\
\log\left(\frac{2}{9}\right) + \frac{1}{2}\log(\eps^{-1}) &< 2^m \log(28)\\
\Omega(\log\log\eps^{-1})&\leq m.\qedhere
\end{align*}
\end{proof}
\begin{lemma}\label{lem:Ei}
    For every $i\in \{1,\ldots , m-1\}$, we have $|E_i|\geq \Omega(n_0)$.
\end{lemma}
\begin{proof}
Let $\tau_i = (2t_i+4)h_{i}$.
For every $j=0,1,\ldots,\lfloor1/h_{i+1}\rfloor-1$, let $C_j = [0,1] \times [jh_{i+1}, (j+1)h_{i+1})$. Subdivide $C_j$ into three axis-aligned rectangles, separated by horizontal lines: The middle rectangle, denoted $M_j$, is of size $1\times \tau_i$, and other two rectangles (above and below $M_j$) are congruent, each of size $1\times 3\tau_i$. We further subdivide $M_j$ into three interior-disjoint axis-aligned rectangles of size $1\times t_ih_{i}$, $1\times 4h_i$, and $1\times t_ih_i$, resp., where the first and the last rectangles are called \emph{upper} and \emph{lower rectangles} of $M_j$; see \Cref{fig:cjmj}.

\begin{figure}[tbh]
    \centering
    \includegraphics[width=.9\textwidth]{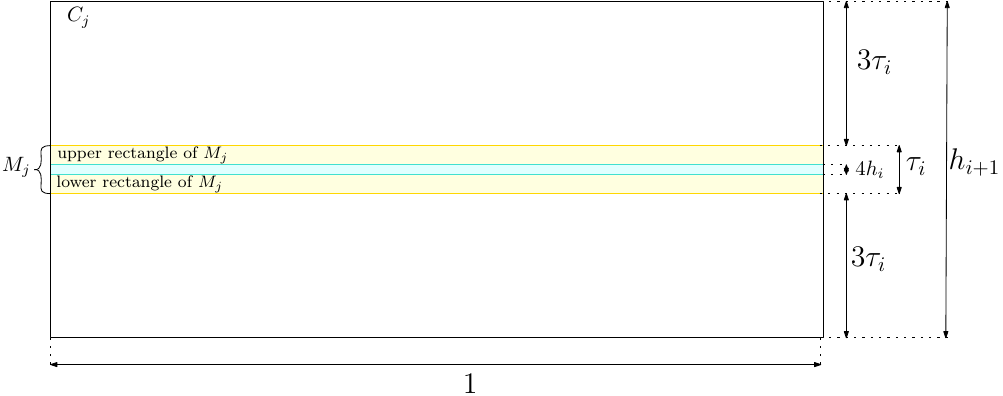}
    \caption{Schematic view of the subdivision of $C_j$ (not to scale)}
    \label{fig:cjmj}
\end{figure}

For every $j=0,1,\ldots,\lfloor1/h_{i+1}\rfloor-1$, 
let $E_{i,j}$ be the set of edges in $E_i$ whose both endpoints are in $C_j$. Note that the sets $E_{i,j}$ are disjoint and we have $ \bigcup_{0 \le j < \lfloor 1/h_{i+1}\rfloor} E_{i,j}\subseteq E_i$, so it suffices to give a lower bound for $|E_{i,j}|$.
We distinguish between two cases: 
\begin{enumerate}
    \item[Case~1: ] At least half of red points in the upper rectangle of $M_j$, or at least half of the blue points in the lower rectangle of $M_j$ are incident to at least one edge in $E_{i, j}$. 
    Recall that the upper and lower rectangles of $M_j$ are separated by an axis-aligned rectangle of size  $1\times 4h_{i}$. It follows that any edge $e$ between the lower and upper rectangles of $M_j$ satisfies $|\mathrm{slope}(e)| \ge 4h_{i}$. Thus, $e \in E_{i, j}$. Since at least a constant fraction of the points in $M_j$ are incident to edges in $E_{i, j}$, then \Cref{obs:height} yields
    \[
    |E_{i,j}| \ge \Omega(|S\cap M_j|) 
    =\Omega\left(\frac{\tau_i}{\sqrt\eps}\right)
    = \Omega\left( \frac{t_ih_i}{\sqrt\eps}\right).
    \]

    \item[Case~2: ] At least half of red points in the upper rectangle of $M_j$, and at least half of the blue points in the lower rectangle of $M_j$ are not incident to any edge in $E_i$.
    Let $R_j^*$ be a maximal set of red points in the upper rectangle of $M_j$ that are not incident to any edge in $E_i$, and the $h_i$-neighborhoods of the points in $R_j^*$ are pairwise disjoint. Similarly, let $B_j^*$ be a maximal set of blue points in the lower rectangle $M_j$ that are not incident to any edge in $E_i$, and the $h_i$-neighborhoods of the points in $B_j^*$ are pairwise disjoint. By \Cref{obs:height}, there are $\lfloor 2h_i/(3\sqrt{\eps})\rfloor$ points in the $h_i$-neighborhood of any point in $M_j$, Consequently, $|R_j^*|\geq \Omega(\tau_i/h_i)=\Omega(t_i)$ and $|B_j^*|\geq \Omega(\tau_i/h_i)=\Omega(t_i)$. 
    
    Consider a point pair $(r,b)\in R_j^*\times B_j^*$.
    Since $r$ (resp., $b$) is in the upper (resp., lower) rectangle of $M_j$, then $|y(r)-y(b)| \ge 4h_{i}$. 
    The $(3+\eps)$-spanner $G$ contains a path $P$ of length $|P|\leq (3+\eps)|rb|$. Since $r\in R_j^*$ and $b\in B_j^*$, then $P$ cannot be a 1-edge path; by \Cref{lem:onethree} it is a 3-edge path that we can denote by $(r,b',r',b)$. Since $r\in R_j^*$ and $b\in B_j^*$, then $|y(r)-y(b')|\leq h_i$ and $|y(r')-y(b)|\leq h_i$. Consequently,
\[|y(b')-y(r')|
\geq |y(r)-y(b)| - |y(r)-y(b')| - |y(r')-y(b)|
\geq 2h_i, 
\]   
and so $b'r'\in E_{i,j}$. Note also, that for all pairs $(r,b)\in R_j^*\times B_j^*$, the edges $r'b'$ are pairwise distinct. Overall, the pairs $(r,b)\in R_j^*\times B_j^*$ yield
$|R_j^*|\cdot |B_j^*|\geq \Omega(t_i^2)=\Omega(t_i\cdot h_i/\sqrt{\eps})$ distinct edges in $E_{i,j}$.

\end{enumerate}

In both cases, we have shown that $|E_{i,j}| \ge \Omega\left(t_ih_i/\sqrt\eps\right)$. Recall that $h_{i+1} = 7(2t_i +4)h_{i}= \Theta(t_ih_i)$. Summation over all $j=1,2,\ldots,\lfloor1/h_{i+1}\rfloor$ yields
  \[
    |E_i| \ge 
    \sum_{j=0}^{\lfloor1/h_{i+1}\rfloor-1} |E_{i,j}|
    \ge \left\lfloor \frac{1}{h_{i+1}}  \right\rfloor \cdot  \Omega\left( \frac{t_ih_i}{\sqrt\eps}\right)
    = \Omega\left( \frac{t_ih_i}{h_{i+1}\sqrt{\eps}}\right)
    = \Omega\left( \frac{1}{\sqrt{\eps}}\right)
    =\Omega(n_0). \qedhere
    \]
\end{proof}

\begin{corollary}\label{cor:basic}
    For every $\eps\in (0,\frac19)$, the basic construction $S=R\cup B$ (defined above) contains $n_0 = 2\, \left\lfloor 1/(3\sqrt\eps)\right\rfloor$ points, and every bichromatic $(3+\eps)$-spanner for $S=R\cup B$ has $\Omega(n_0\cdot \log \log \eps^{-1})$ edges. 
\end{corollary}
\begin{proof}
Let $G=(S,E)$ be a bichromatic $(3+\eps)$-spanner for $S=R\cup B$. The combination of \Cref{lem:manym,lem:Ei} and \Cref{eq:partition2} readily shows that $G$ has $\Omega(\frac{1}{\sqrt\eps} \log\log\eps^{-1})$ edges. 
\end{proof}

\paragraph{General Lower Bound Construction.}
We now show that for $n =\Omega(1/\sqrt\eps)$, every bichromatic $(3+\eps)$-spanner has $\Omega(\log\log\eps^{-1}\cdot n)$ edges (\Cref{thm:LB}). We restate \Cref{thm:LB} for convenience. 

\lowerbound*
\begin{proof}
We may assume w.l.o.g.\ that $\eps\in (0,\frac19)$ and $n>n_0 = 2\, \left\lfloor 1/(3\sqrt\eps)\right\rfloor$.  
Let $m= \lceil n/n_0\rceil$. Arrange $m$ unit squares $Q_1,\ldots , Q_m$ in the plane such that any two squares are at distance at least 10 apart. For $i=1,\ldots ,m-1$, place a congruent copy $S_i$ of the basic construction in the square $Q_i$, each with $n_0$ points, and let $S_m$ be a set of $n-(m-1)n_0$ points of arbitrary colors in $Q_m$. Our point set is $S=\bigcup_{i=1}^m S_i$. Let $G=(S,E)$ be a bichromatic $(3+\eps)$-spanner for $S$. For every $i\in \{1,\ldots , m\}$, define a subgraph $G_i\subseteq G$ on $S_i$: For each red-blue pair $\{r,b\}\subset S_i$, choose a path $P_{r,b}$ of length at most $(3+\eps)|rb|$ in $G$, and let $E(G_i)$ be the union of these paths. Clearly $G_i$ is a $(3+\eps)$-spanner for $S_i$, hence $|E(G_i)|\geq \Omega(n_0\cdot \log \log \eps^{-1})$ by \Cref{cor:basic}. 

We claim that the graphs $G_1,\ldots, G_m$ are vertex-disjoint. Suppose, to the contrary, that $i\neq j$ and $G_i$ and $G_j$ share a vertex $s\in S$. Then $s$ is a vertex of some $r_1b_1$-path $P_1$ of length at most $(3+\eps)|r_1b_1|<\frac{28}{9}\cdot \sqrt{2}<5$ for some red-blue pair $\{r_1,b_1\}\subset S_i$, and also a vertex of some $r_2b_2$-path $P_2$ of length at most $(3+\eps)|r_2b_2|<\frac{28}{9}\cdot \sqrt{2}<5$ for some red-blue pair $\{r_2,b_2\}\subset S_j$,
The union $P_1\cup P_2$ is connected, and contains an $r_1r_2$-path of length at most $|P_1|+|P_2|<10$, contradicting the assumption that the squares $Q_i$ and $Q_j$ are at distance at least 10 apart. 
Now 
\[
    |E(G)|\geq \sum_{i=1}^{m-1} |E(G_i)| 
    \geq (m-1)\cdot  \Omega\left(n_0 \log \log \frac{1}{\eps}\right)  
    = \Omega\left(n \log \log \frac{1}{\eps}\right),
\]
as required.
\end{proof}

\subsection{Tightness of Sparsity Analysis}
\label{ssec:tight}

It is not difficult to see that the sparsity analysis for the basic lower bound construction is tight, up to constant factors. While the general spanner construction in \Cref{sec:2D} uses tilings in $\Theta(\sqrt{1/\eps})$ different orientations, we can use axis-aligned rectangles for the basic lower-bound construction.

\begin{proposition}\label{pro:tight}
    Let $\eps\in (0,\frac13]$; let $R$ (resp., $B$) be a set of $\left\lceil\sqrt{1/\eps}\right\rceil$ equally spaced red (resp., blue) points on the left (resp., right) side of the unit square $[0,1]^2$. Then $S=R\cup B$, with $n=|S|=\Theta(\sqrt{1/\eps})$ admits a bichromatic $(3+\eps)$-spanner with $O(\log \log \eps^{-1} \cdot n)$ edges.
\end{proposition}
\begin{proof}
We construct a subgraph $G$ of $K(R,B)$, and then show that $|E(G)|=O(\log \log \eps^{-1}\cdot n)$ and that $G$ is a bichromatic $(3+\eps)$-spanner.

\paragraph{Spanner Construction.}     
    We cover the unit square $[0,1]^2$ with axis-aligned rectangles at different scales. For $i=0,1,\ldots ,\left\lceil \frac{8}{2^{2^{j}}\sqrt{2\eps/3}}\right\rceil$ and $j=0,1,\ldots , \left\lceil \log \log \frac{1}{\eps}\right\rceil$, we define 
     \[
     C_{i,j}=[0,1]\times \left[i 2^{2^{j}}\cdot \frac{\sqrt{2\eps/3}}{8}, (i+2) 2^{2^{j}}\frac{\sqrt{2\eps/3}}{8}\right].
     \]
     Note that $\mathrm{height}(C_{i,j})=2^{2^{j}}\cdot \sqrt{\eps/24}$, so the height of a rectangle at the smallest scale $j=0$ is $\mathrm{height}(C_{i,0})=\sqrt{\eps/6}$ for all $i$. Note also that every point in $S$ is contained in two rectangles at every scale, hence 
     each point is contained in $\Theta(\log \log \eps^{-1})$ rectangles. 
     In each rectangle $C_{i,j}$, we define a bipartite graph $G_{i,j}$ as follows: 
     Let $r^-_{i,j}$ and $r^+_{i,j}$, resp., be the lowest and highest red point in $R\cap C_{i,j}$; and similarly, let $b^-_{i,j}$ and $b^+_{i,j}$, resp., the lowest and highest blue point in $B\cap C_{i,j}$: Now $G_{i,j}$ is the union of the complete biparite graphs $K(\{r^-_{i,j}, r^+_{i,j}\}, B\cap C_{i,j})$ 
     and $K(R\cap C_{i,j} ,\{b^-_{i,j}, b^+_{i,j}\})$. Let $G$ be the union of  $G_{i,j}$ for all $i$ and $j$, furthermore for every $j\in \{0,1,\ldots ,\left\lceil \log \log \eps^{-1}\right\rceil\}$, add the four edges of the complete bipartite graph $K(\{r^-_{i,j},r^+_{i,j}\}, \{b^-_{i',j}, b^+_{i',j}\})$ if $|i-i'|\leq 2\cdot 2^{2^j}$. This completes the construction of the bipartite graph $G$. 

\paragraph{Sparsity Analysis.}  
     Note that $|E(G_{i,j})|\leq 2\, |S\cap C_{i,j}|$ by construction. Since each point is contained in two rectangles at the same scale and the number of scales is $\Theta(\log \log \eps^{-1})$, then 
     $\sum_{i,j} |E(G_{i,j})|\leq O(\log \log \eps^{-1}\cdot n)$. Furthermore, for every $j\in \{0,1,\ldots , \left\lceil \log \log \eps^{-1}\right\rceil\}$, there are $\Theta(1/(2^{2^j}\sqrt{\eps}))$ rectangles $C_{i,j}$, and we have added $O(2^{2^j})$ edges for each $r^\pm_{i,j}$ and $b^\pm_{i,j}$. The total number of edges is $\Theta(\sqrt{1/\eps})=\Theta(n)$ for each $j$, hence 
     $\Theta(\log \log \eps^{-1}\cdot n)$ over all $j$. 
     Overall, we have $|E(G)|=\Theta(\log \log \eps^{-1}\cdot n)$.
     
\paragraph{Stretch Analysis.}  
     We show that $G$ is a bichromatic $(3+\eps)$-spanner for $S=R\cup B$, that is, 
     for every $(r,b)\in R\times B$, graph $G$ contains an $rb$-path of length at most $(3+\eps)|rb|$. Let $r\in R$ and $b\in B$. Assume w.l.o.g.\ that $y(r)\leq y(b)$.  We distinguish between two cases:

    \begin{itemize}
        \item[Case~1: ] Suppose $|y(r)-y(b)|\leq \sqrt{\eps/6}$. Let $C_{i,j}$ be the smallest rectangle that contains both $r$ and $b$. Since  $|y(r)-y(b)|\leq \sqrt{\eps/6}$, then $j\leq 1$ and $\mathrm{height}(C_{i,j})\leq \sqrt{2\eps/3}$. We have $|rb|\geq 1$, and the length of the path $(r,b^+_{i,j},r^-_{i,j},b)$, using \Cref{eq:P}, is at most 
     \[
        3\, \mathrm{diam}(C_{i,j})
        \leq 3\cdot \sqrt{1+(\sqrt{2\eps/3})^2} 
        \leq 3\left(1+\frac{\eps}{3}\right) 
        \leq (3+\eps)|rb|.
    \]
    
     \item[Case~2:] Suppose $\sqrt{\eps/6}<|y(r)-y(b)|\leq 1$. 
      Let $j\geq 0$ be the integer such that  
     \[
     2^{2^{j}}\cdot\sqrt{\eps/24} < 
     |y(r)-y(b)|
     \leq 2^{2^{j+1}}\cdot\sqrt{\eps/24} .
     \]
    Choose $C_{i,j}$ and $C_{i',j}$ so that they contain the points $r$ and $b$, respectively. Recall that $\mathrm{height}(C_{i,j})=2^{2^j}\cdot \sqrt{\eps/24}$. 
    Consequently, we have 
    \[
    |i-i'|\leq \left\lceil \frac{2\, |y(r)-y(b)|}{\mathrm{height}(C_{i,j})}\right\rceil
    \leq \left\lceil \frac{2\cdot 2^{2^{j+1}}\cdot\sqrt{\eps/24}}{2^{2^j}\cdot \sqrt{\eps/24}}\right\rceil
    = 2\cdot 2^{2^j}.
    \]
    Therefore, $G$ contains the edge $b^+_{i,j}r^-_{i',j}$.
    Since the points in $R$ and $B$ have matching $y$-coordinates, then
    $|y(r^-_{i',j})-y(b^+_{i,j})|\leq |y(r)-y(b)|$. 
    The length of the path $P=(r,b^+_{i,j},r^-_{i',j},b)$ can now be bounded above using \Cref{eq:P} as follows:
    \begin{align*}
    |P|
    &= |rb^+_{i,j}| + |b^+_{i,j}r^-_{i',j}| + |r^-_{i',j}b| \\
    &\leq \mathrm{diam}(C_{i,j}) + |b^+_{i,j}r^-_{i',j}| +\mathrm{diam}(C_{i',j}) \\
    &= 2\, \mathrm{diam}(C_{i,j}) +|b^+_{i,j}r^-_{i',j}|\\
    &= 2\cdot \sqrt{1+(2^{2^{j}}\sqrt{\eps / 24})^2} 
        + \sqrt{1+|y(r^-_{i',j})-y(b^+_{i,j})|^2}\\
    &\leq 2\cdot \sqrt{1+|y(r)-y(b)|^2}
        + \sqrt{1+|y(r)-y(b)|^2}\\
    &= 3\cdot \sqrt{1+|y(r)-y(b)|^2}\\
    &= 3\, |rb| .\qedhere
    \end{align*}
    \end{itemize}
    \end{proof}

\section{Steiner Spanners for Bichromatic Points}
\label{sec:Steiner}

In this section, we present our results for Steiner spanners for the complete bipartite graph $K(R,B)$ on a set of red and blue points with Euclidean edge weights. For the lower bound, we allow non-bipartite Steiner spanners, and for the upper bound we construct a bipartite Steiner spanner with partite sets $\widehat{R}\cup \widehat{B}$ such that $R\subseteq \widehat{R}$ and $B\subseteq \widehat{B}$.

\subsection{Lower Bound for Steiner Spanners}
\label{ssec:SteinerLB}
\SteinerLB*
\begin{proof}
Let $\eps\in (0,\frac19)$ be given. For an even integer $n\geq 2$, we construct a 2-colored set $S$ of $n$ points on the real line: Let $R$ be a set of $n/2$ red points in the interval $(-\eps/8,0)$ and $B$ a set of $n/2$ points in the interval $(1,1+\eps/8)$; and let $S=R\cup B$. Note that for any $r\in R$ and $b\in B$, we have $1< |rb|< 1+\eps/4$. For any two points $u,v\in S$ of same color, the shortest path in $K(R,B)$ has at least two edges, and so $d_{K(R,B)}> 2$. 

Let $H=(\widehat{S},\widehat{E})$ be a Steiner $(2-\eps)$-spanner for $K(R,B)$. Then for every  pair $\{u,v\}\subseteq S$, 
\begin{itemize}
    \item if $u$ and $v$ have different colors, then $d_H(u,v)\leq (2-\eps)|uv|\leq (2-\eps)(1+\eps/4) 
    <2$;
    \item if $u$ and $v$ have the same color and $u\neq v$, then $d_H(u,v)\geq d_{K(R,B)}(u,v) > 2$.
\end{itemize}
Let $\{r_1,b_1\}$ and $\{r_2,b_2\}$ be two disjoint red-blue pairs in $S$; furthermore, let $P_1$ be an $r_1b_1$-path and $P_2$ an $r_2b_2$-path in $H$, each of length less than 2. We claim that $P_1$ and $P_2$ are vertex-disjoint. Indeed, suppose to the contrary that $P_1$ and $P_2$ share a vertex $v$, which decomposes the two paths into subpaths $P_1^-\oplus P_1^+$ and $P_2^-\oplus P_2^+$. Now $P_1^-\oplus P_2^-$ is an $r_1r_2$-path and $P_1^+\oplus P_2^+$ is a $b_1b_2$-path in $H$, each of which must have length more than 2.
However, 
\[ |P_1^-\oplus P_2^-| + |P_1^+\oplus P_2^+| 
=|P_1^-| + |P_2^-| + |P_1^+|+|P_2^+|
=|P_1|+|P_2| <2\cdot 2=4,
\]
which implies $|P_1^-\oplus P_2^-|<2$ or $|P_1^+\oplus P_2^+| <2$, a contradiction. This completes the proof of the claim. Note, however, that if $\{r_1,b_1\}\cap \{r_2,b_2\}\neq \emptyset$, then $P_1$ and $P_2$ may share edges. 

We may assume w.l.o.g.\ (e.g., by symbolic perturbation) that for every red-blue pair $(r,b)\in R\times B$, there is a unique shortest $rb$-path $P_{rb}$ in $H$. 
For every $r\in R$, let $T(r)=\bigcup_{b\in B} P_{rb}$ be a shortest-path tree rooted at $r$ formed by the union of shortest path from $r$ to all points in $B$; and let $E^*(r)$ denote the set of edges of $T(r)$ that lie on two or more shortest paths $P_{rb}$ and $P_{rb'}$ for some $b\neq b'$. Similarly, for every $b\in B$, let $T(b)=\bigcup_{r\in R}P_{rb}$, and let $E^*(b)$ be the set of edges of $T(b)$ that belong to two or more shortest paths $P_{rb}$ and $P_{r'b}$ for some $r\neq r'$. 

We claim that for every red-blue pair $(r,b)\in R\times B$, the shortest path $P_{rb}$ contains an edge $e_{rb}$ that is not in $E^*(r)\cup E^*(b)$ (i.e., $e_{rb}$ does not belong to any other shortest red-blue path). Suppose, to the contrary, that there exists a pair $(r,b)\in R\times B$ such that $P_{rb}\subseteq E^*(r)\cup E^*(b)$. Note that the first edge $e_0$ of $P_{rb}$ is incident to $r$: If $e_0\in E^*(b)$, then $e_0$ is on a shortest path $P_{r'b}$ for some red point $r'\neq r$, consequently a prefix of $P_{r'b}$ contains an $r'r$-path of length at most $|P_{r'b}|<2$: a contradiction. Since $e_0\notin E^*(b)$, then $e_0\in E^*(r)$. Similarly, we can show that the last edge of $P_{rb}$, which is incident to $b$, must be in $E^*(b)$. Let $v$ be the first vertex along $P_{rb}$ incident to an edge in $E^*(b)$. Then $v$ is an interior vertex of $P_{rb}$ incident to an edge $e_1\in E^*(r)$ and edge $e_2\in E^*(b)$. 
This implies that $e_1$ lies on the shortest path $P_{rb'}$ for some $b'\neq b$ and $e_2$ lies on a shortest path $P_{r'b}$ for some $r'\neq r$. In particular, $\{r,b'\}\cap \{r',b\}=\emptyset$ but the shortest paths $P_{rb'}$ and $P_{r'b}$ intersect at vertex $v$: 
A contradiction, which completes the proof of the claim. 

It follows that the function $f: R\times B \to E(H)$, $f(r,b)=e_{rb}$ is injective, consequently $|E(H)|\geq |\{e_{rb}: r\in R, b\in B\}| = |R\times B| = n^2/4$. 
\end{proof}

\subsection{Upper Bound for Steiner Spanners}
\label{ssec:SteinerUB}
\paragraph{Bichromatic Steiner Spanner Construction}
We are given a set $S$ of $n$ bichromatic points in the plane and a parameter $\eps >0$. Let $\mathcal{T}$ be the family of tilings defined in \Cref{subsec:tiling} with the separation constant $\mu = 12$. 

In the construction in \Cref{subsec:construct}, we add edges between the neighborhoods of pairs of tiles $T_{\theta,\lambda}(i,j)$ and $T_{\theta,\lambda}(i+10,j)$. If Steiner points are allowed, we follow a modified version of the construction in \Cref{subsec:construct} and distinguish between four cases based on whether the neighborhoods $N_{\theta,\lambda}(i,j)$ and $N_{\theta,\lambda}(i+13,j)$ are monochromatic or bichromatic.

We begin by initializing $\widehat{S} = S$, let $G=(\widehat{S}, \widehat{E})$ be the empty graph, with the weight function $w:\widehat{E}\to [0,\infty)$. For two points $r,b\in \mathbb{R}^2$ embedded in Euclidean plane, 
the weight of any edge $rb$ will be their Euclidean distance, $w(rb)=|rb|$. However, for any edge incident to Steiner vertices, we need to construct edge weights. Consider four cases. We handle Cases~0, 2, and 3 exactly as in \Cref{subsec:construct} (i.e., no Steiner points are created). Case~1 is replaced by the following version:

    \begin{figure}[tbh]
        \centering
        \includegraphics[width=.9\textwidth]{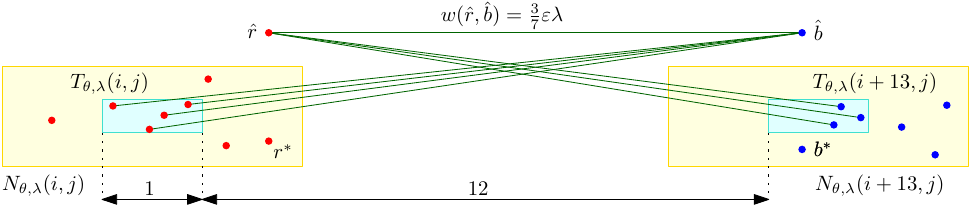}
        \caption{Case~1: Both $N_{\theta,\lambda}(i,j)$ and $N_{\theta,\lambda}(i+13, j)$ are monochromatic}
        \label{fig:monomono-steiner}
    \end{figure}

\begin{itemize}
    \item[Case~1:] Both $N_{\theta,\lambda}(i,j)$ and $N_{\theta,\lambda}(i+13,j)$ are monochromatic, but they contain different colors. Assume w.l.o.g.\ $N_{\theta,\lambda}(i,j) \cap S\subset R$ and $N_{\theta,\lambda}(i+13,j) \cap S \subset B$ (the other case is symmetric); see \Cref{fig:monomono-steiner}. In this case, we add two Steiner points $\hat{r}$ and $\hat{b}$ to $\widehat{S}$, where $\hat{r}$ and $\hat{b}$ are red and blue, respectively. Next, we add an edge between $\hat{r}$ and every point in $T_{\theta,\lambda}(i+13,j) \cap B$, between $\hat{b}$ and every point in $T_{\theta,\lambda}(i,j) \cap R$, as well as the edge $\hat{r}\hat{b}$. In particular, the edges added to $G$ form a \emph{bistar} centered at $\hat{r}$ and $\hat{b}$. Lastly, 
    we assign the edge weights $w(\hat{r}\hat{b}) = \frac37\eps\lambda$, $w(r\hat{b}) = |rb^*|$ for all $r \in T_{\theta,\lambda}(i,j)$, and $w(b\hat{r}) = |br^*|$ for all $b \in T_{\theta,\lambda}(i+13,j)$, where $r^*$ is a rightmost red point in $N_{\theta,\lambda}(i,j)$ and $b^*$ is a leftmost blue point in $N_{\theta,\lambda}(i+13,j)$.
\end{itemize}

This completes the construction of the edge-weighted bipartite Steiner graph $G=(\widehat{S}, \widehat{E})$. \Cref{thm:Steiner} below shows that $G$ is a bichromatic Steiner $(2+\eps)$-spanner of $K(R,B)$. It is important to note that for a Steiner spanner, beyond an upper bound on the stretch, we also need to show that the shortest path distance is noncontracting. The sparsity analysis is identical to the generalization to multichromatic point sets in \Cref{ssec:multichromatic}. 

\SteinerUB*
\begin{proof}
We show that for every point pair $\{s,t\}\in \binom{S}{2}$, the shortest-path distance in $G$ satisfies 
$d_{K(R,B)}(s,t)\leq d_G(s,t)\leq (2+\eps)d_{K(R,B)}(s,t)$. For the upper bound $d_G(r,b)\leq (2+\eps)d_{K(R,B)}(r,b)$, it is enough to consider adjacent pairs of points in $K(R,B)$, where the distance in $K(R,B)$ equals the Euclidean distance. In particular, for every red-blue pair $(r,b)\in R\times B$, we need to show that $d_G(r,b)\leq (2+\eps)|rb|$. For non-contraction, $d_{K(R,B)}(s,t)\leq d_G(s,t)$, we consider all pairs.

\paragraph{Upper Bound.}
    We proceed by induction on the length of the bichromatic edges in $K(R,B)$. 
    In particular, let $D=\{ |rb|: (r,b)\in R\times B\}$ be the set of distances between red-blue pairs in $S$; and let $0=d_0<d_1<\ldots <d_m$ be elements of $\{0\}\cup D$ in increasing order. 
    
    We prove, by induction on $\ell=0,1,\ldots , m$, that for all $(r,b)\in R\times B$ with $|rb|\leq d_\ell$, the graph $G$ contains an $rb$-path of length at most $(2+\eps)|rb|$. In the base case, where $\ell=0$ and $d_\ell=0$, the claim vacuously holds, as there is no red-blue pair $(r,b)$ with $|rb|=0$. 

    In the induction step, let $d_\ell\in D$ and assume that for all $(r',b')\in R\times B$, if $|r'b'|<d_\ell$, then $G$ contains an $r'b'$-path of length at most $(2+\eps)|r'b'|$.  
    Let $(r,b)\in R\times B$ with $|rb|=d_\ell$. By Lemma~\ref{lemma:spacing}, there exist $\theta\in \Phi$, $\lambda \in \Lambda$ and $i,j \in \mathbb{Z}$ such that $T_{\theta,\lambda}(i,j)$ and $T_{\theta,\lambda}(i+13,j)$ each contain one of $r$ and $b$; or $T_{\theta,\lambda}'(i,j)$ and $T_{\theta,\lambda}'(i+13,j)$ each contain one of $r$ and $b$. 

   To simplify our argument, we may assume w.l.o.g. that $\theta = 0$, $\lambda = 1$, $r\in T_{\theta,\lambda}(i,j)$ and $b \in T_{\theta,\lambda}(i+13,j)$. Under these assumptions, $N_{\theta, \lambda}(i,j) = N(i,j)$, which is an axis-aligned rectangle of width $3$ and height $3\delta$.  We distinguish between three cases:

   \begin{itemize}
       \item[Case~1:] Both $N_{\theta,\lambda}(i,j)$ and $N_{\theta,\lambda}(i+13,j)$ are monochromatic; see \Cref{fig:monomono-steiner}. By construction (Case~1), $G$ contains the $rb$-path $P = (r,\hat{b},\hat{r}, b)$. Since $r^*$ is a rightmost red point in $N(i,j)$ and $b^*$ is a leftmost blue point in $N(i+13,j)$, then we have $x(r) \le x(r^*) < x(b^*) \le x(b)$, where $x(r)$, $x(r^*)$, $x(b^*)$ and $x(b)$, resp., denote the $x$-coordinate of $r$, $r^*$, $b^*$ and $b$. 
       Then, the edge weight $w(r\hat{b}) = |rb^*|$ can be bounded above by the Pythagorean Theorem and \Cref{eq:P}. It follows that
       \[
       w(r\hat{b}) = |rb^*| \le \sqrt{(x(b) -x(r))^2 + (2\delta)^2} \le x(b) - x(r) +\frac{2\eps}{7} \le |rb| + \frac{2\eps}{7}.
       \]
       By a similar process, we can attain the same upper bound for $w(\hat{r}b) \le |rb|+2\eps/7$. Combined with $w(\hat{b}\hat{r}) = 3\eps / 7$ and $|rb|\ge 12$, we have that the length of the $rb$-path $P$ is bounded above as follows
       \[
        w(P) = w(r\hat{b}) + w(\hat{b}\hat{r}) + w(\hat{r} b) \le 2|rb| + \eps \le (2+\eps)|rb|.
       \]
       Thus, $G$ contains an $rb$-path of length at most $(2+\eps)|rb|$. (Note that we did not use the induction hypothesis in Case~1.)

       \item[Case~2:] One of $N(i,j)$ and $N(i+13,j)$ is monochromatic and the other is bichromatic. This case is identical to Case~2 in \Cref{lem:stretch} and the upper bound on the stretch factor follows similarly from \Cref{eq:case2} with the new separation constant $\mu = 12$ and the new induction hypothesis. Thus, $G$ contains an $rb$-path of length at most $(2+\eps)|rb|$.

       \item[Case~3:] Both $N(i,j)$ and $N(i+13,j)$ are multichromatic. This case is identical to Case~3 in \Cref{lem:stretch} and the upper bound on the stretch factor follows similarly from \Cref{eq:case3} with the new separation constant $\mu = 12$ and the new induction hypothesis. Thus, $G$ contains an $rb$-path of length at most $(2+\eps)|rb|$.
   \end{itemize}

    Overall, the bichromatic Steiner graph $G$ is a $(2+\eps)$-spanner of $K(R,B)$.

\paragraph{Noncontraction.} Let $s,t\in S$. We show that $d_G(s,t)\geq d_{K(R,B)}(s,t)$. Let $P$ be a shortest $st$-path in $G$. Let $P=P_1\oplus P_2\oplus \ldots \oplus P_k$ be a decomposition of $P$ into maximal subpaths in which all interior vertices are Steiner points. It is enough to show that each path $P_i$ is noncontracting, that is, if the endpoints of $P_i$ are $s_i,t_i\in S$, then $w(P_i)\geq |s_it_i|$. By construction, every subpath $P_i$ is a single edge, or a subgraph of a bistar added in Case~1 of the construction. Every edge of $K(R,B)$ has the same weight in $K(R,B)$ and $G$, so it is enough to consider subpaths in a bistar. Assume that$P_i$ is a subpath in a bistar with star centers $\hat{r}$ and $\hat{b}$ created in Case~1 of the construction. There are two cases to consider:
   \begin{itemize}
       \item[Case~1:] The endpoints of $P_i$ have the same color. Assume w.l.o.g.\ that both endpoints of $P_i$ are red, $r_1,r_2\in R\cap T_{\theta, \lambda}(i,j)$. Then $P_i=(r_1,\hat{b},r_2)$. By the choice of the edge weights, we have 
       \[
       w(P_i) = w(r_1,\hat{b})+ w(\hat{b},r_2)= |r_1b^*|+|b^* r_2| =|(r_1, b^*,r_2)| \geq d_{K(R,B)} (r_1,r_2).
       \]
       \item[Case~2:]  The endpoints of $P_i$ have different colors. Let the endpoints of $P_i$ be $r\in R\cap T_{\theta, \lambda}(i,j)$ and $b\in B\cap T_{\theta, \lambda}(i+13,j)$. We may assume w.l.o.g.\ that $T_{\theta,\lambda}(i,j)$ has unit width (i.e., $\lambda=1$). Then 
       $P_i=(r,\hat{b},\hat{r},b)$, and we have 
       \[
       w(P_i) > w(r,\hat{b}) + w(\hat{r},b) = |rb^*| + |r^* b| \geq 11+11 = 22 > |rb|.
       \]
 \end{itemize}

\paragraph{Sparsity Analysis.} In each of the four cases, we added the same number of edges as in the construction of the non-Steiner bichromatic spanners, so we can use the sparsity analysis in \Cref{subsec:analysis2} with separation constant $\mu=12$. The charging schemes for edges of Type~1 (\Cref{lem:type1-pre}) and Type~2 (\Cref{lem:type2-pre}) show that the number of edges is $O(n)$.

\paragraph{Running Time Analysis.} The construction can be implemented in $O(\sqrt{1/\eps}\cdot n\log n)$ time. The construction of the tilings, and identifying the pairs of nonempty tiles is identical to the algorithm for non-Steiner spanners in \Cref{ssec:runtime}, with separation constant $\mu=12$ in place of $\mu =9$. In Case~1, finding the points $r^*$, $b^*$, and constructing a bistar centered at Steiner points $\hat{r}$ and $\hat{b}$ takes the same time as in Case~1 for the non-Steiner spanner. 
\end{proof}

\section{Bichromatic Spanners on the Line}
\label{sec:1D}

In this section, we study spanners for $n$ bichromatic points on the real line $\mathbb{R}$. The lower bound construction in \cite{BoseCCMMS09} shows that, for every $\eps>0$, there is a bichromatic point set such that a $(3-\eps)$-spanner must contain all edges of the complete bipartite graph $K(R,B)$; this construction consists of two monochromatic clusters of diameter $\frac{\eps}{3}$ at unit distance apart, and can be realized already in 1-dimension. We show that there is a 3-spanner with at most $2n-3$ edges (\Cref{prop:3}; furthermore, the minimum spanning tree of $K(R,B)$ (denoted $\mathrm{minBST}(S)$) is a 7-spanner, and stretch bound 7 is the best possible (\Cref{prop:1D}). 

We start by introducing some notation. 
Let $S=R\cup B$ be a set of $n$ bichromatic points on a line. Assume that $S=\{s_1,\ldots , s_n\}$, where the points are labeled in linear order along the line. Partition $S$ into maximal subsets of consecutive points of the same color $S=\bigcup_{i=1}^k S_i$; we call these subsets \emph{blocks}. Clearly, the blocks span disjoint intervals along the line, which are also ordered linearly. For any nonempty subset $S'\subseteq S$, denote by $\ell(S')$ and $r(S')$ the leftmost and rightmost point in $S'$, respectively. 

\begin{proposition}\label{prop:3}
    For a set $S=B\cup R$ of $n$ bichromatic points on a line, there is a 3-spanner with at most $2n-3$ edges that can be constructed in $O(n\log n)$ time.
\end{proposition}
\begin{proof}
We construct a subgraph $H$ of $K(R,B)$ as follows. 
Consider the partition of $S$ into blocks, $S=\bigcup_{i=1}^k S_i$. 
For every $i=1,\ldots , k$, and for every point $s\in S_i$, add edges $s r(S_{i-1})$ and $s \ell(S_{i+1})$ to $H$ (if these edges exist). In particular, 
we add only one edge for every $s\in S_1\cup S_k$, and two edges for all other vertices.  Note also that the $k-1$ edges of the form $r(S_i)\ell(S_{i+1})$ are counted twice. Since $|S_1|\geq 1$, $|S_k|\geq 1$, and $k\geq 2$, then the total number of edges in $H$ is $2n-(|S_1|+|S_k|)-(k-1)\leq 2n-2-1=2n-3$.

Next we show that $H$ is a 3-spanner. Let $a,b\in S$ be two points of opposite colors, and assume w.l.o.g. that $a\in S_i$, $b\in S_j$ for some $1\leq i<j\leq k$. Then $H$ contains the path 
\[
    P=(a,\ell(S_{i+1}),r(S_i), r(S_{i+1}), \ldots , 
    r(S_{j-1}), b).
\]
The first three edges of the path $P$ intersect in the interval, $r(S_i), \ell(S_{i+1})$. All subsequent edges are pairwise interior disjoint (from each other and from the first three edges). This implies that every point on the line is contained in at most 3 edges of $P$, which immediately implies that $|P|\leq 3\, |ab|$, as required. 
\end{proof}

\paragraph{Bichromatic Minimum Spanning Tree.}
Bandyapadhyay et al.~\cite{BandyapadhyayBB21} showed that if $n$ bichromatic points are already sorted along the line, then $\mathrm{minBST}(S)$ can be computed in $O(n)$ time. 
(Biniaz et al.~\cite{BiniazBEMMS18} showed how to compute $\mathrm{minBST}(S)$ in $O(n\log n)$ in the plane.) 

We start with a characterization of $\mathrm{minBST}(S)$, based on the observations in~\cite{BandyapadhyayBB21}. Consider a block $S_i$, $1\leq i\leq k$, defined above. For every point $s\in S_i$, the closest point of the opposite color is either $r(S_{i-1})$ or $\ell(S_{i+1})$. We define a partition $S_i=S_i^-\cup S_i^+$, where $S_i^-$ (resp., $S_i^+$) is the set of points in $S_i$ for which the closest point of the opposite color is in $S_{i-1}$ (resp., $S_{i+1}$). Note that $S_1^-=\emptyset$ and $S_k^+=\emptyset$, but $S_i^-$ or $S_i^+$ may be empty for any $1<i<k$.

Recall that if $G=(V,E,w)$ is an edge-weighted connected graph with positive edge weights, then for every nontrivial cut $V=V_1\cup V_2$, the MST of $G$ contains a minimum-weight edge between $V_1$ and $V_2$. We apply this observation for nontrivial cuts $S=\{s_1,\ldots , s_{j}\}\cup \{s_{j+1},\ldots , s_n\}$ for $j=1,\ldots, n-1$ in the complete bipartite graph $K(R,B)$. Specifically, if $s_j$ and $s_{j+1}$ are in different blocks, then the shortest bichromatic edge across the cut is $s_js_{j+1}$; in this case $s_j=r(S_i)$ and $s_{j+1}=\ell(S_{i+1})$ for some $1\leq i\leq k$. 
Now assume that both $s_j$ and $s_{j+1}$ are in the same block $S_i$. If $s_j,s_{j+1}\in S_i^-$ (resp., $s_j,s_{j+1}\in S_i^+$), then $r(S_{i-1}) s_{j+1}$ (resp., $s_j \ell(S_{i+1})$) is the shortest bichromatic edge across the cut. Finally, if $s_j\in S_i^-$ and $s_{j+1}\in S_i^+$, then the shortest bichromatic edge across the cut is the shorter of two candidates: $r(S_{i-1}) s_{j+1}$ or $s_j\ell(S_{i+1})$. Note that we considered $n-1$ nontrivial cuts, and obtained $n-1$ distinct edges: Let $H$ be the graph formed by these edges. 
Since $H$ has $n-1$ edges and it is a subgraph of $\mathrm{minBST}(S)$, then 
$H= \mathrm{minBST}(S)$. 

We give a tight bound on the stretch of the bichromatic MST on a line.

\begin{proposition}\label{prop:1D}
    For a set $S=B\cup R$ of $n$ bichromatic points on a line, $\mathrm{minBST}(S)$ is a 7-spanner; and for every $\eps>0$, there is a bichromatic point set $S\subset \mathbb{R}$ such that the stretch of $\mathrm{minBST}(S)$ is at least $7-\eps$.
\end{proposition}
\begin{proof}
It is enough to show that for any pair of points $a,b$ in consecutive blocks, $\mathrm{minBST}(S)$ contains a path of length at most $7\, |ab|$. Indeed, if $a\in S_i$ and $b\in S_j$, where $i<j$, then we can first consider the 
path $P=(a,\ell(S_{i+1}),\ldots ,\ell(S_{j-1}), b)$ in $K(R,B)$. 
Note that $|P|=|ab|$, and every edge of $P$ is between points in consecutive blocks. 
By replacing every edge $e\subseteq P$ by a path in $\mathrm{minBST}$ of length at most $7\, |e|$, we obtain an $ab$-path in $\mathrm{minBST}$ of length at most 
$\sum_{e\subseteq P} 7\, |e| =7\, |P|=7\, |ab|$.

We may assume w.l.o.g.\ that $a\in S_i$ and $b\in S_{i+1}$. We claim that $\mathrm{minBST}(S)$ contains a path $P_1$ from $a$ to $\ell(S_{i+1})$ of length $|P_1|\leq 3\, |ab|$; and symmetrically a path $P_2$ from $r(S_i)$ to $b$ of length $|P_2|\leq 3\, |ab|$. Then the concatenation $P_1\oplus\ell(S_{i+1})r(S_i)\oplus P_2$ is an $ab$-path in $\mathrm{minBST}(S)$, and its length is at most $|P_1| + |\ell(S_{i+1}) r(S_i)| +|P_2|\leq (3+1+3)|ab|=7\, |ab|$, as required. It remains to prove the claim. We distinguish between two cases:
\begin{itemize}
\item \textbf{Case~1}: \textbf{$a\in S_i^+$ or  $a=r(S_i^-)$ and $r(S_i^-)\ell(S_{i+1})\in \mathrm{minBST}(S)$.} In this case, $P_1$ is a single edge path, and $|P_1|=| a\, \ell(S_{i+1})|\leq |ab|$. 
\item \textbf{Case~2}: \textbf{$a\in S_i^-$ and if $a=r(S_i^-)$, then $r(S_i^-)\ell(S_{i+1})\notin \mathrm{minBST}(S)$.} The latter condition implies that $r(S_{i-1})\ell(S_i^+)\in \mathrm{minBST}(S)$ because 
$r(S_{i-1})\ell(S_i^+)$ is shorter than $r(S_i^-)\ell(S_{i+1})$. In particular, we have $|r(S_{i-1})\ell(S_i^+)|\leq |r(S_i^-)\ell(S_{i+1})|\leq |ab|$.
In this case, let $P_1=(a, r(S_{i-1}), \ell(S_i^+), \ell(S_{i+1}))$. Now each of the three edges of $P_1$ has length at most $|ab|$, and so $|P_1|\leq 3\, |ab|$, as required.
\end{itemize}
This confirms the claim in both cases; and completes the proof of the upper bound.

\begin{figure}[tbh]
        \centering
        \includegraphics[width=.65\textwidth]{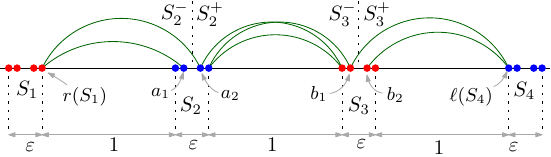}
        \caption{Lower bound construction}
        \label{fig:minBSTs}
\end{figure}

For a lower bound consider a bichromatic point set with 4 blocks, $S=\bigcup_{i=1}^4S_i$, each of 4 points, such that each block spans an interval of length $\eps>0$ for a sufficiently small $\eps>0$, and the distance between any two consecutive blocks is 1; see \Cref{fig:minBSTs}.  Let $a_1,a_2$ be the two interior points in $S_2$ such that $a_1\in S_2^-$, $a_2\in S_2^+$, and $|r(S_1)a_2|<|a_1\ell(S_3)|$. 
Symmetrically, let $b_1,b_2$ be the two interior points in $S_3$ such that $b_1\in S_3^-$, $b_2\in S_3^+$, and $|b_1\ell(S_4)|<|r(S_2)b_2|$. Then $a_1$ has degree 1 in $\mathrm{minBST}(S)$, and its only neighbor is $r(S_1)$. Similarly, the only neighbor of $b_2$ is $\ell(S_4)$. The set $S_2^+\cup S_3^-$ induces a bistar centered at $r(S_2)$ and $\ell(S_3)$: A bistar on 4 vertices forms  a spanning path on $S_2^+\cup S_3^-$ between $a_2$ and $b_1$. Overall, the $a_1b_2$-path in $\mathrm{minBST}(S)$ is
\[ 
    P=(a_1, r(S_1),a_2,\ell(S_3), r(S_2),b_1, \ell(S_4),b_2).
\]
We have $|a_1b_2|\leq 1+2\eps$, but every edge of $P$ has length at least $1$. Consequently the stretch of $\mathrm{minBST}(S)$ is at least $|P|/|a_1b_2|\geq 7/(1+2\eps)$, which tends to 7 as $\eps\to 0$. 
\end{proof}

\section{Conclusions}
\label{sec:conclusion}

We conclude with a few open problems. 
What is the tight dependence on $\eps>0$ for the size of a bichromatic $(3+\eps)$-spanner and a bichromatic Steiner $(2+\eps)$-spanner for red-blue points in the plane? 
On the line, we have constructed a 3-spanner with at most $2n-3$ edges, that is, no dependence on $\eps$ (\Cref{prop:3}). 
In the plane, we have proven that the sparsity of a bichromatic $(3+\eps)$-spanner is between $O(\sqrt{1/\eps})$ and $\Omega(\log \log \eps^{-1})$ (\Cref{thm:main,thm:LB}). 
With off-metric Steiner points, we constructed bichromatic Steiner $(2+\eps)$-spanner with $O(\sqrt{1/\eps}\cdot n)$ (\Cref{thm:Steiner}), but we do not know whether the dependence on $\eps$ is necessary: Is a bichromatic Steiner 2-spanner possible with $O(n)$ edges?

The \emph{lightness} of a spanner $H$ for an edge-weighted graph $G$ is the ratio $w(H)/w(\mathrm{MST}(G))$ of the weight of the spanner to the weight of a minimum spanning tree of $G$. On bichromatic points on the line, we have shown (\Cref{prop:1D}) that there is a 7-spanner of lightness 1. It remains an open problem to obtain upper bounds for stretch $3\leq t<7$ on the line, and for $t>3$ in the plane. 
 
Our results give existential bounds on the size of a bichromatic spanner. The corresponding \emph{optimization problems} remain open: For a given bichromatic point set $S=R\cup B$ and stretch $t\geq 1$, find a $t$-spanner of $K(R,B)$ of minimum size or minimum weight (note that for a given instance, stretch $1\leq t<3$ may be feasible). We conjecture that these problems are NP-hard already in Euclidean plane. Finding good approximations is left for future research. For monochromatic spanners in the plane, these optimization problems are known to be NP-hard~\cite{GKKKM10,KK06}, and bicriteria approximation algorithms were  obtained only recently~\cite{le2024towards}. However, for bichromatic points, the optimization problem is already open for $n$ red-blue points on the line with stretch $1<t<7$. 

\paragraph*{Acknowledgments.}
We are grateful to Sujoy Bhore for initial conversations on the structure of bichromatic spanners on the real line, and to Jonathan Conroy for suggesting the problems on off-metric Steiner points for bichromatic spanners.

\bibliographystyle{alphaurl}
\bibliography{spanners}

@article{BoseCCMMS09,
  author       = {Prosenjit Bose and
                  Paz Carmi and
                  Mathieu Couture and
                  Anil Maheshwari and
                  Pat Morin and
                  Michiel H. M. Smid},
  title        = {Spanners of Complete $k$-Partite Geometric Graphs},
  journal      = {{SIAM} J. Comput.},
  volume       = {38},
  number       = {5},
  pages        = {1803--1820},
  year         = {2009},
  doi          = {10.1137/070707130}
}

@article{BoseCCMSZ09,
  author       = {Prosenjit Bose and
                  Paz Carmi and
                  Mathieu Couture and
                  Anil Maheshwari and
                  Michiel H. M. Smid and
                  Norbert Zeh},
  title        = {Geometric spanners with small chromatic number},
  journal      = {Comput. Geom.},
  volume       = {42},
  number       = {2},
  pages        = {134--146},
  year         = {2009},
  doi          = {10.1016/J.COMGEO.2008.04.003}
}

@article{Abu-AffashBCM21,
  author       = {A. Karim Abu{-}Affash and
                  Sujoy Bhore and
                  Paz Carmi and
                  Joseph S. B. Mitchell},
  title        = {Planar bichromatic bottleneck spanning trees},
  journal      = {J. Comput. Geom.},
  volume       = {12},
  number       = {1},
  pages        = {109--127},
  year         = {2021},
  doi          = {10.20382/JOCG.V12I1A5}
}

@article{AgarwalES91,
  author       = {Pankaj K. Agarwal and
                  Herbert Edelsbrunner and
                  Otfried Schwarzkopf},
  title        = {Euclidean Minimum Spanning Trees and Bichromatic Closest Pairs},
  journal      = {Discret. Comput. Geom.},
  volume       = {6},
  pages        = {407--422},
  year         = {1991},
  doi          = {10.1007/BF02574698}
}

@article{AkitayaBDKST25,
  author       = {Hugo A. Akitaya and
                  Ahmad Biniaz and
                  Erik D. Demaine and
                  Linda Kleist and
                  Frederick Stock and
                  Csaba D. T{\'{o}}th},
  title        = {Minimum Plane Bichromatic Spanning Trees},
  journal      = {{ACM} Trans. Algorithms},
  volume       = {21},
  number       = {4},
  pages        = {48:1--48:14},
  year         = {2025},
  doi          = {10.1145/3747591}
}

@article{BorgeltKLLMSV09,
  author       = {Magdalene G. Borgelt and
                  {Marc J. van} Kreveld and
                  Maarten L{\"{o}}ffler and
                  Jun Luo and
                  Damian Merrick and
                  Rodrigo I. Silveira and
                  Mostafa Vahedi},
  title        = {Planar bichromatic minimum spanning trees},
  journal      = {J. Discrete Algorithms},
  volume       = {7},
  number       = {4},
  pages        = {469--478},
  year         = {2009},
  doi          = {10.1016/J.JDA.2008.08.001}
}

@book{BergCKO08,
  author       = {{Mark de} Berg and
                  Otfried Cheong and
                  {Marc J. van} Kreveld and
                  Mark H. Overmars},
  title        = {Computational Geometry: {A}lgorithms and Applications},
  edition={3rd},
  publisher    = {Springer},
  year         = {2008},
  doi          = {10.1007/978-3-540-77974-2}
}

@article{BiniazBMS18,
  author       = {Ahmad Biniaz and
                  Prosenjit Bose and
                  Anil Maheshwari and
                  Michiel H. M. Smid},
  title        = {Plane Bichromatic Trees of Low Degree},
  journal      = {Discret. Comput. Geom.},
  volume       = {59},
  number       = {4},
  pages        = {864--885},
  year         = {2018},
  doi          = {10.1007/S00454-017-9881-Z}
}

@inproceedings{BiniazMS14a,
  author       = {Ahmad Biniaz and
                  Anil Maheshwari and
                  Michiel H. M. Smid},
  title        = {Bottleneck Bichromatic Plane Matching of Points},
  booktitle    = {Proc.\ 26th Canadian Conference on Computational Geometry ({CCCG})},
  publisher    = {Carleton University, Ottawa, ON},
  year         = {2014},
  url          = {http://www.cccg.ca/proceedings/2014/papers/paper63.pdf}
}

@article{BandyapadhyayBB21,
  author       = {Sayan Bandyapadhyay and
                  Aritra Banik and
                  Sujoy Bhore and
                  Martin N{\"{o}}llenburg},
  title        = {Geometric planar networks on bichromatic collinear points},
  journal      = {Theor. Comput. Sci.},
  volume       = {895},
  pages        = {124--136},
  year         = {2021},
  doi          = {10.1016/J.TCS.2021.09.035}
}

@article{BiniazBEMMS18,
  author       = {Ahmad Biniaz and
                  Prosenjit Bose and
                  David Eppstein and
                  Anil Maheshwari and
                  Pat Morin and
                  Michiel H. M. Smid},
  title        = {Spanning Trees in Multipartite Geometric Graphs},
  journal      = {Algorithmica},
  volume       = {80},
  number       = {11},
  pages        = {3177--3191},
  year         = {2018},
  doi          = {10.1007/S00453-017-0375-4}
}

@article{Abu-AffashBCC19,
  author       = {A. Karim Abu{-}Affash and
                  Sujoy Bhore and
                  Paz Carmi and
                  Dibyayan Chakraborty},
  title        = {Bottleneck bichromatic full {S}teiner trees},
  journal      = {Inf. Process. Lett.},
  volume       = {142},
  pages        = {14--19},
  year         = {2019},
  doi          = {10.1016/J.IPL.2018.10.003}
}

@article{BaltzS05,
  author       = {Andreas Baltz and
                  Anand Srivastav},
  title        = {Approximation algorithms for the {E}uclidean bipartite {TSP}},
  journal      = {Oper. Res. Lett.},
  volume       = {33},
  number       = {4},
  pages        = {403--410},
  year         = {2005},
  doi          = {10.1016/J.ORL.2004.08.002}
}

@article{BhoreNTW26,
  author       = {Sujoy Bhore and
                  Martin N{\"{o}}llenburg and
                  Csaba D. T{\'{o}}th and
                  Jules Wulms},
  title        = {Fully Dynamic Maximum Independent Sets of Disks in Polylogarithmic
                  Update Time},
  journal      = {Discret. Comput. Geom.},
  volume       = {75},
  number       = {2},
  pages        = {391--430},
  year         = {2026},
  doi          = {10.1007/S00454-025-00793-8}
}

@article{BeregHKTTZ25,
  author       = {Sergey Bereg and
                  Yuya Higashikawa and
                  Naoki Katoh and
                  Junichi Teruyama and
                  Yuki Tokuni and
                  Binhai Zhu},
  title        = {Constructing red-black spanners for mixed-charging vehicular networks},
  journal      = {Theor. Comput. Sci.},
  volume       = {1023},
  pages        = {114932},
  year         = {2025},
  doi          = {10.1016/J.TCS.2024.114932}
}

@inproceedings{KarpR14,
  author       = {Jeremy Karp and
                  R. Ravi},
  _editor       = {Klaus Jansen and Jos{\'{e}} D. P. Rolim and Nikhil R. Devanur and Cristopher Moore},
  title        = {A $9/7$-Approximation Algorithm for Graphic {TSP} in Cubic Bipartite
                  Graphs},
  booktitle    = {Approximation, Randomization, and Combinatorial Optimization. Algorithms
                  and Techniques ({APPROX/RANDOM})},
  series       = {LIPIcs},
  pages        = {284--296},
  publisher    = {Schloss Dagstuhl},
  year         = {2014},
  doi          = {10.4230/LIPICS.APPROX-RANDOM.2014.284},
}

@inproceedings{Nederlof20,
  author       = {Jesper Nederlof},
  _editor       = {Konstantin Makarychev and Yury Makarychev and Madhur Tulsiani and                  Gautam Kamath and Julia Chuzhoy},
  title        = {Bipartite {TSP} in $o(1.9999^n)$ time, assuming quadratic time matrix multiplication},
  booktitle    = {Proc.\ 52nd {ACM} Symposium on Theory of Computing ({STOC})},
  pages        = {40--53},
  _publisher    = {{ACM}},
  year         = {2020},
  doi          = {10.1145/3357713.3384264}
}

@article{MantasPSW25,
  author       = {Ioannis Mantas and
                  Evanthia Papadopoulou and
                  Rodrigo I. Silveira and
                  Zeyu Wang},
  title        = {The Farthest Color {V}oronoi Diagram in the Plane},
  journal      = {Algorithmica},
  volume       = {87},
  number       = {10},
  pages        = {1393--1419},
  year         = {2025},
  doi          = {10.1007/S00453-025-01311-1}
}

@inproceedings{BOP25,
  author       = {Sang Won Bae and
                  Nicolau Oliver and
                  Evanthia Papadopoulou},
  _editor       = {Oswin Aichholzer and Haitao Wang},
  title        = {Higher-Order Color {V}oronoi Diagrams and the Colorful {C}larkson-{S}hor
                  Framework},
  booktitle    = {Proc. 41st Symposium on Computational Geometry (SoCG)},
  series       = {LIPIcs},
  pages        = {12:1--12:19},
  publisher    = {Schloss Dagstuhl},
  year         = {2025},
  doi          = {10.4230/LIPICS.SOCG.2025.12}
}

@article{AbamH12,
  author       = {Mohammad Ali Abam and
                  Sariel Har{-}Peled},
  title        = {New constructions of {SSPD}s and their applications},
  journal      = {Comput. Geom.},
  volume       = {45},
  number       = {5-6},
  pages        = {200--214},
  year         = {2012},
  doi          = {10.1016/J.COMGEO.2011.12.003}
}

@book{HarPeled11,
series={Mathematical Surveys and Monographs},
volume={173},
title={Geometric Approximation Algorithms},
author={Sariel Har-Peled},
year={2011},
publisher={AMS},
address={Boston, MA},
doi={10.1090/surv/173}
}

@inproceedings{Clarkson87,
  author       = {Kenneth L. Clarkson},
  _editor       = {Alfred V. Aho},
  title        = {Approximation Algorithms for Shortest Path Motion Planning},
  booktitle    = {Proc.\ 19th {ACM} Symposium on Theory of Computing (STOC)},
  pages        = {56--65},
  _publisher    = {{ACM}},
  year         = {1987},
  doi          = {10.1145/28395.28402}
}

@article{CallahanK95,
  author       = {Paul B. Callahan and
                  S. Rao Kosaraju},
  title        = {A Decomposition of Multidimensional Point Sets with Applications to
                  $k$-Nearest-Neighbors and $n$-Body Potential Fields},
  journal      = {J. {ACM}},
  volume       = {42},
  number       = {1},
  pages        = {67--90},
  year         = {1995},
  doi          = {10.1145/200836.200853}
}

@book{NS07,
  title = {Geometric Spanner Networks},
  ISBN = {9780511546884},
  DOI = {10.1017/cbo9780511546884},
  publisher = {Cambridge University Press},
  author = {Narasimhan,  Giri and Smid,  Michiel},
  year = {2007}
}

@inproceedings{le2024towards,
  author       = {Hung Le and
                  Shay Solomon and
                  Cuong Than and
                  Csaba D. T{\'{o}}th and
                  Tianyi Zhang},
  title        = {Towards Instance-Optimal Euclidean Spanners},
  booktitle    = {Proc.\ 65th {IEEE} Symposium on Foundations of Computer Science (FOCS)},
  pages        = {1579--1609},
  year         = {2024},
  doi          = {10.1109/FOCS61266.2024.00099}
}

@InProceedings{Keil88,
  Title                    = {Approximating the Complete {E}uclidean Graph},
  Author                   = {J.~Mark Keil},
  Booktitle                = {Proc.\ 1st Scandinavian Workshop on Algorithm Theory {(SWAT})},
  Year                     = {1988},
  Pages                    = {208--213},
  series       = {LNCS},
  volume       = {318},
  publisher    = {Springer},
  doi          = {10.1007/3-540-19487-8_23}
}

@inproceedings{KK06,
  title={Computing geometric minimum-dilation graphs is {NP}-hard},
  author={Klein, Rolf and Kutz, Martin},
  booktitle={Proc.\ 14th International Symposium on Graph Drawing},
  pages={196--207},
 series={LNCS},
 volume={4372},
  publisher={Springer},
 DOI = {10.1007/978-3-540-70904-6_20},
  year={2006}
}

@article{GKKKM10,
  title = {COMPUTING GEOMETRIC MINIMUM-DILATION GRAPHS IS {NP}-HARD},
  volume = {20},
  ISSN = {1793-6357},
  DOI = {10.1142/s0218195910003244},
  number = {2},
  journal = {International Journal of Computational Geometry \& Applications},
   author={Giannopoulos, Panos and Klein, Rolf and Knauer, Christian and Kutz, Martin and Marx, D{\'a}niel},
  year = {2010},
  pages = {147--173}
}

@inproceedings{FischerH05,
  author       = {John Fischer and
                  Sariel Har{-}Peled},
  title        = {Dynamic Well-Separated Pair Decomposition Made Easy},
  booktitle    = {Proc.\ 17th Canadian Conf. Comput. Geom. (CCCG)},
  address={Windsor, ON},
  pages        = {235--238},
  year         = {2005},
  url          = {http://www.cccg.ca/proceedings/2005/32.pdf}
}

@inproceedings{KrauthgamerL04,
  author       = {Robert Krauthgamer and
                  James R. Lee},
  _editor       = {J. Ian Munro},
  title        = {Navigating nets: simple algorithms for proximity search},
  booktitle    = {Proc.\ 15th {ACM-SIAM} Symposium on Discrete
                  Algorithms ({SODA})},
  pages        = {798--807},
  year         = {2004},
  url          = {http://dlacm.org/citation.cfm?id=982792.982913}
}

@article{HarPeledM06,
  author       = {Sariel Har{-}Peled and
                  Manor Mendel},
  title        = {Fast Construction of Nets in Low-Dimensional Metrics and Their Applications},
  journal      = {{SIAM} J. Comput.},
  volume       = {35},
  number       = {5},
  pages        = {1148--1184},
  year         = {2006},
  doi          = {10.1137/S0097539704446281}
}

@inproceedings{Talwar04,
  author       = {Kunal Talwar},
  _editor       = {L{\'{a}}szl{\'{o}} Babai},
  title        = {Bypassing the embedding: algorithms for low dimensional metrics},
  booktitle    = {Proc.\ 36th {ACM} Symposium on Theory of Computing (STOC)},
  pages        = {281--290},
  year         = {2004},
  doi          = {10.1145/1007352.1007399}
}

@article{AbamCFS13,
  author       = {Mohammad Ali Abam and
                  Paz Carmi and
                  Mohammad Farshi and
                  Michiel H. M. Smid},
  title        = {On the power of the semi-separated pair decomposition},
  journal      = {Comput. Geom.},
  volume       = {46},
  number       = {6},
  pages        = {631--639},
  year         = {2013},
  doi          = {10.1016/J.COMGEO.2013.02.003}
}

\end{document}